\documentclass{article}

\usepackage[margin=1.25in]{geometry}
\usepackage{amsmath}
\usepackage{amssymb}
\usepackage{graphics}
\usepackage{graphicx}
\usepackage{amsthm}
\usepackage[blocks]{authblk}

\usepackage{times}
\usepackage{bm}
\usepackage{natbib}

\usepackage{bbm}

\makeatletter

\usepackage{url}
\usepackage{mathtools}

\def\lstRomanletters{a,b,c,d,e,f,g,h,i,j,k,l,m,n,o,p,q,r,s,t,u,v,w,x,y,z}
\def\lstRomanLETTERS{A,B,C,D,E,F,G,H,I,J,K,L,M,N,O,P,Q,R,S,T,U,V,W,X,Y,Z}
\def\lstRomanSca{a,b,c,d,e,f,g,h,i,j,k,l,m,n,o,p,q,r,s,t,u,v,w,x,y,z,%
	A,B,C,D,E,F,G,H,I,J,K,L,M,N,O,P,Q,R,S,T,U,V,W,X,Y,Z}
\def\lstRomanBold{avec,bvec,cvec,dvec,evec,fvec,gvec,hvec,ivec,jvec,kvec,lvec,mvec,nvec,ovec,pvec,qvec,rvec,svec,tvec,uvec,vvec,wvec,xvec,yvec,zvec,%
	Amat,Bmat,Cmat,Dmat,Emat,Fmat,Gmat,Hmat,Imat,Jmat,Kmat,Lmat,Mmat,Nmat,Omat,Pmat,Qmat,Rmat,Smat,Tmat,Umat,Vmat,Wmat,Xmat,Ymat,Zmat}
\def\lstscript{mat,cal,scr,bb,frak}
\def\MkRomanVec#1.{\expandafter\def\csname#1vec\endcsname{\expandafter\mathsf{#1}}}
\def\MkRomanFonts#1.#2.{\expandafter\def\csname#1#2\endcsname{\csname math#2\endcsname{#1}}}

\@for\i:=\lstRomanletters\do{\expandafter \MkRomanVec \i.}
\@for\i:=\lstRomanLETTERS\do{\@for\j:=\lstscript\do{\edef\temp{\i.\j.} \expandafter\MkRomanFonts \temp}}

\def\lstGreekletters{alpha,beta,gamma,delta,epsilon,zeta,eta,theta,iota,kappa,lambda,mu,nu,xi,omicron,pi,rho,sigma,tau,upsilon,phi,chi,psi,omega}
\def\lstGreekLETTERS{Gamma,Delta,Theta,Lambda,Xi,Pi,Sigma,Upsilon,Phi,Psi,Omega}
\def\lstGreekAll{alpha,beta,gamma,delta,epsilon,zeta,eta,theta,iota,kappa,lambda,mu,nu,xi,omicron,pi,rho,sigma,tau,upsilon,phi,chi,psi,omega,alphavec,betavec,gammavec,deltavec,epsilonvec,zetavec,etavec,thetavec,iotavec,kappavec,lambdavec,muvec,nuvec,xivec,omicronvec,pivec,rhovec,sigmavec,tauvec,upsilonvec,phivec,chivec,psivec,omegavec,Gamma,Delta,Theta,Lambda,Xi,Pi,Sigma,Upsilon,Phi,Psi,Omega,Gammamat,Deltamat,Thetamat,Lambdamat,Ximat,Pimat,Sigmamat,Upsilonmat,Phimat,Psimat,Omegamat}
\def\MkGreekVec#1.{\expandafter\def\csname#1vec\endcsname{\expandafter\boldsymbol{\csname#1\endcsname}}}
\def\MkGreekMat#1.{\expandafter\def\csname#1mat\endcsname{\expandafter\boldsymbol{\csname#1\endcsname}}}
\def\MkGreekAdorn#1.#2.{\expandafter\def\csname#1#2\endcsname{\csname math#2\endcsname{\csname#1\endcsname}}}
\def\MkRomanAdorn#1.#2.{\expandafter\def\csname#1#2\endcsname{\csname math#2\endcsname{#1}}}
\def\MkGreekStar#1.#2.{\expandafter\def\csname#1#2\endcsname{\csname math#2\endcsname{\csname#1\endcsname}^{\csname super#2\endcsname}}}
\def\MkRomanStar#1.#2.{\expandafter\def\csname#1#2\endcsname{\csname math#2\endcsname{#1}^{\csname super#2\endcsname}}}

\def\lstadornments{hat,tilde,bar}
\def\lstsuperadornments{star}

\@for\i:=\lstGreekletters\do{\expandafter \MkGreekVec \i.}
\@for\i:=\lstGreekLETTERS\do{\expandafter \MkGreekMat \i.}
\@for\i:=\lstRomanSca\do{\@for\j:=\lstadornments\do{\edef\temp{\i.\j.} \expandafter\MkRomanAdorn \temp}}
\@for\i:=\lstRomanBold\do{\@for\j:=\lstadornments\do{\edef\temp{\i.\j.} \expandafter\MkGreekAdorn \temp}}
\@for\i:=\lstGreekAll\do{\@for\j:=\lstadornments\do{\edef\temp{\i.\j.} \expandafter\MkGreekAdorn \temp}}
\@for\i:=\lstRomanSca\do{\edef\temp{\i.star.} \expandafter\MkRomanStar \temp}
\@for\i:=\lstRomanBold\do{\edef\temp{\i.star.} \expandafter\MkGreekStar \temp}
\@for\i:=\lstGreekAll\do{\edef\temp{\i.star.} \expandafter\MkGreekStar \temp}

\def\lstnumbers{zero.0,one.1,two.2,three.3,four.4,five.5,six.6,seven.7,eight.8,nine.9}

\def\MkNums#1.#2.{\expandafter\def\csname#1\endcsname{#2}}
\def\MkRomanCons#1.#2.#3.{\expandafter\def\csname#1#2\endcsname{#1_{#3}}}
\def\MkGreekCons#1.#2.#3.{\expandafter\def\csname#1#2\endcsname{\csname#1\endcsname_{#3}}}
\@for\i:=\lstnumbers\do{\expandafter\MkNums \i.}
\@for\i:=\lstRomanSca\do{\@for\j:=\lstnumbers\do{\edef\temp{\i.\j.} \expandafter\MkRomanCons \temp}}
\@for\i:=\lstRomanBold\do{\@for\j:=\lstnumbers\do{\edef\temp{\i.\j.} \expandafter\MkGreekCons \temp}}
\@for\i:=\lstGreekAll\do{\@for\j:=\lstnumbers\do{\edef\temp{\i.\j.} \expandafter\MkGreekCons \temp}}
\@for\j:=\lstnumbers\do{\edef\temp{Vert.\j.} \expandafter\MkGreekCons \temp}

\@for\i:=\lstGreekAll\do{\@for\j:=\lstadornments\do{\@for\k:=\lstnumbers\do{\edef\temp{\i\j.\k.} \expandafter\MkGreekCons \temp}}}
\@for\i:=\lstGreekAll\do{\@for\j:=\lstsuperadornments\do{\@for\k:=\lstnumbers\do{\edef\temp{\i\j.\k.} \expandafter\MkGreekCons \temp}}}
\@for\i:=\lstRomanSca\do{\@for\j:=\lstadornments\do{\@for\k:=\lstnumbers\do{\edef\temp{\i\j.\k.} \expandafter\MkGreekCons \temp}}}
\@for\i:=\lstRomanSca\do{\@for\j:=\lstsuperadornments\do{\@for\k:=\lstnumbers\do{\edef\temp{\i\j.\k.} \expandafter\MkGreekCons \temp}}}
\@for\i:=\lstRomanBold\do{\@for\j:=\lstadornments\do{\@for\k:=\lstnumbers\do{\edef\temp{\i\j.\k.} \expandafter\MkGreekCons \temp}}}
\@for\i:=\lstRomanBold\do{\@for\j:=\lstsuperadornments\do{\@for\k:=\lstnumbers\do{\edef\temp{\i\j.\k.} \expandafter\MkGreekCons \temp}}}

\let\epsilonold\epsilon
\let\epsilon\varepsilon
\let\varepsilon\epsilonold
\let\phiold\phi
\let\phi\varphi
\let\varphi\phiold

\def\comp{\mathsf{C}}

\DeclareMathOperator{\pr}{pr}
\DeclareMathOperator\Ber{Ber}

\makeatother

\def\Bkl{B_{k\ell}}

\def\bkln{b_{k\ell, n}}
\def\bklN{b_{k\ell, N}}

\def\bklnone{b_{k\ell, n-1}}

\def\btildeklN{\btilde_{k\ell,N}}

\def\pkl{p_{k\ell}}
\def\qkl{q_{k\ell}}
\def\qkld{\qkl^{(d)}}

\def\qtildekl{\qtilde_{k\ell}}
\def\qhatkl{\qhat_{k\ell}}

\def\Bkld{\Bkl^{d}}
\def\pr{P}

\def\P{P}
\def\Q{Q}
\def\Qd{Q^{(d)}}

\def\fac#1{(#1)!}
\DeclareMathOperator\kl{KL}

\makeatletter
\def\lstadornments{hat,tilde,bar,check}

\@for\i:=\lstRomanSca\do{\@for\j:=\lstadornments\do{\edef\temp{\i.\j.} \expandafter\MkRomanAdorn \temp}}
\let\inold\in
\def\MkRomanAlphaCons#1.#2.#3.#4.{\expandafter\def\csname#1#4#2\endcsname{#1_{#4-#3}}}
\def\MkRomanAlpha#1.{\expandafter\def\csname#1n\endcsname{#1_{n}} \expandafter\def\csname#1N\endcsname{#1_{N}}}
\def\MkGreekAlphaCons#1.#2.#3.#4.{\expandafter\def\csname#1#4#2\endcsname{\csname#1\endcsname_{#4-#3}}}
\def\MkGreekAlpha#1.{\expandafter\def\csname#1n\endcsname{\csname#1\endcsname_{n}} \expandafter\def\csname#1N\endcsname{\csname#1\endcsname_{N}}}
\@for\i:=\lstRomanSca\do{\edef\temp{\i.} \expandafter\MkRomanAlpha \temp}
\@for\i:=\lstGreekAll\do{\edef\temp{\i.} \expandafter\MkGreekAlpha \temp}
\@for\i:=\lstRomanSca\do{\@for\j:=\lstnumbers\do{\@for\k:=\lstRomanSca\do{\edef\temp{\i.\j.\k.} \expandafter\MkRomanAlphaCons \temp}}}
\@for\i:=\lstRomanSca\do{\@for\j:=\lstadornments\do{\@for\k:=\lstnumbers\do{\@for\l:=\lstRomanSca\do{\edef\temp{\i\j.\k.\l.} \expandafter\MkGreekAlphaCons \temp}}}}
\@for\i:=\lstRomanSca\do{\@for\j:=\lstadornments\do{\edef\temp{\i\j.} \expandafter\MkGreekAlpha \temp}}
\let\in\inold
\makeatother

\def\ONE{\mathbbm{1}}

\theoremstyle{plain}
\newtheorem{theorem}{Theorem}
\newtheorem{lemma}{Lemma}
\newtheorem{corollary}{Corollary}[theorem]
\newtheorem{proposition}{Proposition}
\theoremstyle{definition}

\theoremstyle{remark}
\newtheorem{remark}{Remark}

\theoremstyle{plain}

\DeclareMathOperator{\Leb}{Leb}

\title{A Rank-Based Sequential Test of Independence}
\author[1]{Alexander Henzi\thanks{Supported in part by European Research Council (ERC) under European Union’s Horizon 2020 research and innovation
programme, grant agreement No. 786461}}
\author[1]{Michael Law\thanks{Supported in part by NSF Grant DMS-2203012.}}
\affil[1]{Seminar for Statistics, ETH Z\"urich}

\begin{document}

\maketitle

\begin{abstract}
    We consider the problem of independence testing for two univariate random variables in a sequential setting.  By leveraging recent developments on safe, anytime-valid inference, we propose a test with time-uniform type I error control and derive explicit bounds on the finite sample performance of the test.   We demonstrate the empirical performance of the procedure in comparison to existing sequential and non-sequential independence tests.  Furthermore, since the proposed test is distribution free under the null hypothesis, we empirically simulate the gap due to Ville's inequality --- the supermartingale analogue of Markov's inequality --- that is commonly applied to control type I error in anytime-valid inference, and apply this to construct a truncated sequential test.
\end{abstract}

\section{Introduction}

Let $(X_n, Y_n) \in \mathbb{R}\times\mathbb{R}$, $n \in \mathbb{N}$, be a stream of independent and identically distributed (iid) random variables with unknown distribution $\P$, and denote by $(X,Y)$ a generic pair from this distribution. In this article, we consider the classical problem of testing independence of $X$ and $Y$; that is, the null hypothesis is that for all Borel measurable sets $A, B \subseteq \mathbb{R}$,
\[
H_0\colon \P(X \in A, Y \in B) = \P(X\in A)\P(Y\in B).
\]
This problem has been studied for almost a century and there exists a vast literature on fixed sample size procedures for testing $H_0$, which we do not attempt to summarize. However, only little work has been done on testing independence sequentially. By a sequential test, we mean a decision rule $\varphi\colon \bigcup_{n=1}^{\infty} \mathbb{R}^{2n} \rightarrow \{0, 1\}$, with $\varphi = 1$ for ``reject $H_0$'' and $\varphi = 0$ for ``do not reject $H_0$'', which satisfies
\begin{equation} \label{eq:type_i_error}
	\P\{\varphi(X_1, Y_1, \dots, X_n, Y_n) = 1 \text{ for some } n \in \mathbb{N}\} \leq \alpha
\end{equation}
for a prescribed significance level $\alpha \in (0,1)$ under $H_0$. As can be seen from the definition, compared to non-sequential tests, a sequential test allows us to monitor the data stream $(X_n, Y_n)_{n \in \mathbb{N}}$ continuously over time, while still controlling the overall probability of a false rejection, without having to fix a sample size $n$ in advance.

Sequential tests of independence have been studied only recently by \citet[Appendix A.2]{Balsubramani2016}, \citet[Section 5.2.2]{Shekhar2021}, \citet{Podkopaev2022}, and \citet{Podkopaev2023}; a more extensive literature review is in Section \ref{sec:literature_review}. These approaches belong to the area of safe, anytime-valid inference, which is reviewed in detail by \citet{Ramdas2022game}. We refer the interested reader to this article for an introduction to the area and a general overview of the definitions and methods that we apply throughout this paper.

A common tool for constructing sequential tests, applied in the articles cited above, are test martingales. A test martingale is a non-negative $(\mathcal{G}_n)_{n\in\mathbb{N}}$-supermartingale $(M_n)_{n \in \mathbb{N}}$ with initial value $M_0 = 1$, where $(\mathcal{G}_n)_{n\in\mathbb{N}}$ is a certain filtration. By Ville's inequality, $(M_n)_{n\in\mathbb{N}}$ satisfies
\begin{equation} \label{eq:ville}
	\P(M_n \geq 1/\alpha \text{ for some } n \in \mathbb{N}) \leq \alpha, \ \alpha \in (0,1),
\end{equation}
under $H_0$, so that the function $\varphi = \ONE(M_n \geq 1/\alpha)$ is a sequential test with type I error probability $\alpha$. More precisely, the above cited articles consider the data in pairs $Z_n = \{(X_{2n-1}, Y_{2n-1}), (X_{2n}, Y_{2n})\}$ and build test martingales $(M_n)_{n\in\mathbb{N}}$ of the form
\begin{equation} \label{eq:betting_martingale}
	M_N = \prod_{n=1}^N [1+\lambda_n \{g_n(Z_n) - g_n(\tilde{Z}_n)\}],
\end{equation}
where $\tilde{Z}_n = \{(X_{2n-1}, Y_{2n}), (X_{2n}, Y_{2n-1})\}$ swaps the $Y$-variable in the pairs, and $\lambda_n$ and $g_n$ are weights and test functions such that $\lambda_n \{g_n(Z_n) - g_n(\tilde{Z}_n)\} \geq -1$ for all $n \in \mathbb{N}$. 

The first contribution of this article is the result that with the original data filtration $\mathcal{F}_n = \sigma(X_i, Y_i, i = 1, \dots, n)$, any $(\mathcal{F}_n)_{n\in\mathbb{N}}$- test martingale $(M_n)_{n \in \mathbb{N}}$ for $H_0$ is non-increasing. Hence $M_n \leq 1$ for all $n \in \mathbb{N}$ and the corresponding test has no power to detect violations of the null hypothesis. The same phenomenon also occurs for tests of other large nonparametric hypothesis classes, namely, exchangeability of binary sequences \citep{Ramdas2022exchangeability} and log-concavity \citep{Gangrade2023}. There are two strategies to bypass this problem. One is to construct so-called e-processes, which satisfy weaker conditions than test martingales but still yield anytime-valid tests; see \citet{Ramdas2022game} for more details. The second approach is to replace $(\mathcal{F}_n)_{n \in \mathbb{N}}$ by a coarser filtration. The above mentioned articles on independence testing take $Z_n = \{(X_{2n-1}, Y_{2n-1}), (X_{2n}, Y_{2n})\}$ as observation units, and hence follow the second approach with $(\mathcal{G}_n)_{n\in\mathbb{N}}$ defined as $\mathcal{G}_n = \mathcal{F}_n$ for even $n$ and $\mathcal{G}_n = \mathcal{F}_{n-1} \subset \mathcal{F}_{n}$ for odd $n$.

In this article, we propose a new rank-based sequential test of independence. Instead of working with the original observations $(X_n, Y_n)$, we replace them by their normalized sequential ranks
\begin{align} \label{eq:sequential_ranks}
	\hat{F}_n(X_n) = \frac{1}{n}\sum_{i=1}^n \ONE(X_i \leq X_n), \quad \hat{G}_n(Y_n) = \frac{1}{n}\sum_{i=1}^n \ONE(Y_i \leq Y_n).
\end{align}
With suitable randomization, this transformation allows us to reduce the composite null hypothesis of independence to the simple hypothesis that the randomized ranks are iid uniform random variables on $[0,1]^2$ --- a point null hypothesis for which one can construct powerful test martingales. The resulting test martingales are adapted to the filtration generated by the sequential ranks and, hence, we follow a different strategy than the existing methods. As a beneficial side effect, since our method is based on ranks, it is invariant under strictly increasing transformations of $(X_n)_{n \in \mathbb{N}}$ and $(Y_n)_{n \in \mathbb{N}}$ and also distribution free.

To construct our test martingale, we discretize the unit square $[0,1]^2$ into equally sized bins and then assess whether the bin frequencies follow a uniform multinomial distribution. Binning approaches are an established strategy for independence testing; see for example the recent articles by \citet{Ma2019}, \citet{Zhang2019}, and the references therein. There are several reasons why we consider this approach suitable for sequential independence testing. From a theoretical point of view, the resulting test martingales have mathematically tractable forms, which make it possible to derive uniform bounds on the type II error of the test. From a practical point of view, the approach is computationally attractive since the computation of the test martingale up to time $n$ only requires $\mathcal{O}\{n\log(n)\}$ operations, and each update is of complexity $\mathcal{O}\{\log(n)\}$, if implemented properly. Also, the test is interpretable in the sense that it comes along with an estimate of bin frequencies on discretizations of $[0,1]^2$, from which one can visualize which regions deviate from the uniform distribution and explain the dependence between $X$ and $Y$.

The following notation is used throughout the article. For elements $x_1, \dots, x_n$ of a set $\mathcal{X}$, we let $\xvec_{n} = (x_{1}, \dots, x_{n}) \in \Xcal^{n}$, with no distinction between row and column vectors; we analogously write $\Xmat_n = (X_1, \dots, X_n)$ for random vectors. The concatenation of $\uvec_n \in \mathcal{X}^n$ and $\vvec_m \in \mathcal{X}^m$ is denoted by $(\uvec_n, \vvec_m)$. We use $\mathbb{E}_P$ for the expectation operator under the distribution $P$.  The indicator function is denoted by $\mathbbm{1}(\cdot)$. We use the convention $0\log(0) = 0$. The Lebesgue measure of $A\subseteq \mathbb{R}^n$ is denoted by $\mathrm{Leb}(A)$, where the dimension of the space is apparent from the context. For a sequence $(X_n, Y_n)_{n \in \mathbb{N}}$ of iid pairs $(X_n, Y_n)$ with distribution $P$, we let $(X,Y)$ denote a generic pair with that distribution. The notation $\mathcal{U}$ is used for the uniform distribution on a bounded or finite set, and $\kl(P\|Q) = \int p(x)\log\{p(x)/q(x)\} \, \mu(dx)$ is the Kullback-Leibler (KL) divergence between two distributions $P, Q$ dominated by a measure $\mu$.

\section{Related literature} \label{sec:literature_review}
This work belongs to the area of safe, anytime-valid inference, which is reviewed in great detail by \citet{Ramdas2022game}; see also \citet{Shafer2011}, \citet{Grunwald2019} and \citet{Vovk2021} for the general methodology. Test martingales, which are a main tool for constructing sequential tests, have close connections to game theoretic probability \citep{ShaferVovk2019}, in that they can be interpreted as a sequence of bets against the null hypothesis, each of which has an expected payoff less or equal to $1$ if the null hypothesis is true. The value of the test martingale is then the accumulated capital if all payoffs are invested in the bets as new observations arrive. In this article we are mainly interested in type I and II error of our test, but the validity of test martingales holds in more generality. Namely, if $(M_n)_{n\in\mathbb{N}}$ is a $(\mathcal{G}_n)_{n\in\mathbb{N}}$ test martingale, then by Doob's martingale convergence theorem,
\begin{equation} \label{eq:stopped_martingale}
	\mathbb{E}_P(M_{\tau}) \leq 1
\end{equation}
for any $(\mathcal{G}_n)_{n\in\mathbb{N}}$ stopping time $\tau$ under $H_0$. The stopped process is also referred to as an e-variable, and satisfies $P(M_{\tau} \geq 1/\alpha) \leq \alpha$ and, hence, type I error control by Markov's inequality. Tests based on e-variables are also called e-tests. We are primarily interested in the aggressive stopping rule $\tau = \inf(n\in\mathbb{N}\colon M_n \geq 1/\alpha)$ for a level $\alpha$ test, but in many practical settings experiments are stopped early due to other criteria, both external and data-dependent, and hence the validity under arbitrary stopping rule is an important property of a statistical test.

As mentioned above, the articles by \citet{Balsubramani2016}, \citet{Shekhar2021}, \citet{Podkopaev2022}, \citet{Podkopaev2023} develop methods for independence testing in sequential settings. Other authors treat more general hypothesis testing problems, or independence testing in special cases. \citet{Gruenwald2022modelx}, \citet{Shaer2022}, and \citet{Duan2022} develop anytime-valid tests of conditional independence under the model-X assumption, which, for non-conditional independence tests, requires one of the marginal distributions of $(X,Y)$ to be known. The case of binary observations, e.g., contingency tables, is treated by \citet{Turner2021}, and extended to stratified count data by \citet{Turner2023}.

The special case of testing independence of Gaussian random variables, which is equivalent to testing correlations, has attracted most attention in the sequential testing literature. This line of work started, to the best of our knowledge, with \citet{Cox1952}, and includes many follow-up articles \citep{Choi1971, Kowalski1971, Pradhan1975, Kollerstrom1979, Kocherlakota1986}. In the literature on anytime-valid testing, the problem is covered by the recently developed group-invariant tests by \citet{Perez2022}.

The only work applying sequential ranks for independence testing we are aware of is by \citet{Choi1973}, but their test is computationally infeasible and one has to rely on approximations under which its validity is not clear. In other sequential testing problems sequential ranks have been applied more extensively; see \citet{Kalina2017} for an overview.

\section{Rank-based tests}

\subsection{An impossibility result}
We begin with the result that any non-negative $(\mathcal{F}_n)_{n\in\mathbb{N}}$-supermartingale for $H_0$ almost surely has non-increasing paths. The proof of this theorem, as well as all other proofs, can be found in Appendix \ref{sec:proofs}.

\begin{theorem} \label{thm:no_supermartingale}
    Let $(X_n, Y_{n})_{n \in \mathbb{N}}$ be iid and define the filtration $\mathcal{F}_{n} = \sigma(X_{i}, Y_{i}, i = 1, \dots, n)$.
    \begin{itemize} 
        \item[(i)] If $M_{n}$ is an $\mathcal{F}_{n}$-supermartingale for every $P$ satisfying $H_{0}$, then $M_{n} \geq M_{n+1}$ everywhere.

        \item[(ii)] If $M_{n}$ is an $\mathcal{F}_{n}$-supermartingale for every $P$ satisfying $H_{0}$ and absolutely continuous with respect to Lebesgue measure, then $P'(M_{n} \geq M_{n+1}) = 1$ for every $P'$ absolutely continuous with respect to Lebesgue measure.
    \end{itemize}
\end{theorem}

The second part of the theorem shows that even if we restrict the family of distributions to Lebesgue continuous distributions satisfying $H_0$, there still only exist trivial supermartingales, in the sense that they decrease almost surely if the data follows any Lebesgue continuous distribution. The above impossibility result motivates us to consider martingales with a coarser filtration than $\sigma(X_i, Y_i, i = 1, \dots,n)$, $n \in \mathbb{N}$, which we introduce in the next section.

\subsection{Constructing the test martingale} \label{sec:construction}
Before going into the details of our method, we emphasize that although our test is based on sequential ranks and hence most suitable for distributions with continuous marginals, all results on validity and power hold for any type of marginal distributions of $(X,Y)$.

As explained in the Introduction, to construct a test martingale, we transform the observations $(X_n, Y_n)$ to their sequential ranks $\hat{F}_n(X_n)$, $\hat{G}_n(Y_n)$, where $\hat{F}_n$ and $\hat{G}_n$ are the empirical cumulative distribution functions (CDFs) of $X_1, \dots, X_n$ and $Y_1, \dots, Y_n$, respectively. Since the support of $\hat{F}_n(X_n)$, $\hat{G}_n(Y_n)$ changes with $n$, it is convenient to randomize the ranks as follows,
\begin{equation} \label{eq:randomized_ranks}
	R_n = U_n\hat{F}_n(X_n) + (1-U_n) \hat{F}_n(X_n-), \quad S_n = V_n\hat{G}_n(Y_n) + (1-V_n) \hat{G}_n(Y_n-),
\end{equation}
with independent uniform variables $(U_n)_{n \in \mathbb{N}}$, $(V_n)_{n \in \mathbb{N}}$ on $(0,1)$, independent of all other random quantities. The following proposition, part (i) of which is by \citet{BarndorffNielsen1963}, summarizes useful properties of the sequential ranks. 

\begin{proposition} \label{prop:sequential_ranks}
	Let $(X,Y) \sim \pr$ with marginal CDFs $F$ and $G$, respectively.
	\begin{itemize}
		\item[(i)] If $F$ is continuous, then the sequential ranks $\hat{F}_n(X_n)$, $n = 1, \dots, N$, are independent with
		\[
		\pr\left\{\hat{F}_n(X_n) = \frac{k}{n}\right\} = \frac{1}{n}, \ k = 1, \dots, n, \ n = 1, \dots, N.
		\]
		for all $N \in \mathbb{N}$. The same result holds for $G$ and $\{\hat{G}_n(Y_n)\}_{n \in \mathbb{N}}$.
		\item[(ii)] The randomized ranks $(R_n)_{n \in \mathbb{N}}$, $(S_n)_{n \in \mathbb{N}}$ are independent, but not necessarily mutually independent, sequences uniformly distributed on $[0,1]$.
		\item[(iii)] If $X$ and $Y$ are independent, then $(R_n, S_n)_{n\in\mathbb{N}}$ are independent and uniform on $[0,1]^2$.
	\end{itemize}
\end{proposition}

By part (iii) of the above proposition, the transformation to randomized sequential ranks reduces the composite hypothesis of independence to the simpler problem of testing whether $(R_n, S_n)_{n \in \mathbb{N}}$ are uniformly distributed on $[0,1]^2$. Admissible test martingales for this hypothesis, in the sense of \citet{Ramdas2020}, are likelihood ratio processes
\begin{equation} \label{eq:likelihood_ratio_process}
	M_N = \prod_{n=1}^N f_n(R_n, S_n), \ N \in \mathbb{N},
\end{equation}
where $(f_n)_{n \in \mathbb{N}}$ are densities on $[0,1]^2$, and the likelihood ratio is relative to the uniform density, which is constant $1$. The functions $f_n$ may, and in practice do, depend on $\Rmat_{n-1}$, $\Smat_{n-1}$; that is, $f_n(r,s) = g_n(r, s, \Rmat_{n-1}, \Smat_{n-1})$ for a measurable function $g_n$ so that $(r, s) \mapsto g_n(r, s, \Rmat_{n-1}, \Smat_{n-1})$ is a density for all values of $\Rmat_{n-1}, \Smat_{n-1}$. We indicate the dependence of $f_n$ on $\Rmat_{n-1}, \Smat_{n-1}$ only with the subscript $n$.

A process $(M_n)_{n \in \mathbb{N}}$ of the form \eqref{eq:likelihood_ratio_process} is a test martingale and hence yields a valid sequential test for any predictable choice of $(f_n)_{n\in\mathbb{N}}$. The challenge is to find a method for constructing $(f_n)_{n\in\mathbb{N}}$ such that the corresponding test has power against a large class of alternatives. Our strategy relies on a simple but effective binning approach. For an integer $d > 0$, we partition $[0,1]^2$ into $d^2$ grid cells $\Bkld = (k/d, (k+1)/d] \times (\ell/d, (\ell+1)/d]$ and $B^d_{k0} = (k/d, (k+1)/d] \times [0, 1/d]$, $B^d_{0\ell} = [0,1/d] \times (\ell, (\ell+1)/d]$ for $k,\ell = 1, \dots, d-1$; the convention that the cells are left open and right closed is not essential since values on the boundaries occur with probability zero due to the randomization in $R_n, S_n$. We then construct $f_n$ with a bivariate histogram estimator; for each cell $\Bkld$, we count the frequency of observations in $\Bkld$ up to time $n-1$, and divide it by the cell size $d^{-2}$. Following \citet{Hall1988} and \citet{Yu1992}, we include an initial count of $1$ in each cell to avoid density estimates of exactly zero, so that the estimator becomes
\begin{align} \label{eq:density_estimator}
	f_n^{(d)}(r,s) = d^2\prod_{k,\ell=0}^{d-1} \left[\frac{1+\sum_{j=1}^{n-1} \ONE\big\{(R_j, S_j) \in \Bkld\big\}}{d^2+n-1}\right]^{\ONE\{(r,s) \in \Bkld\}}.
\end{align}
With this choice of $(f_n^{(d)})_{n\in\mathbb{N}}$, the test martingale $\mathcal{M}^{d} = (M_n^{(d)})_{n\in\mathbb{N}}$ equals
\begin{align} \label{eq:test_martingale_d}
	\MN^{(d)} 
    = \prod_{n=1}^{N} \prod_{k,\ell=0}^{d-1} d^2 \Big( \frac{\bklnone + 1}{n-1+d^{2}} \Big)^{\ONE\{(\Rn,\Sn) \in \Bkld\} }
	= \frac{d^{2N}\fac{d^{2}-1}}{\fac{N-1+d^{2}}} \prod_{k,\ell=0}^{d-1} \fac{\bklN},
\end{align}
where $\bkln = \sum_{j=1}^n \ONE\{(R_j, S_j) \in \Bkld\}$ is the number of observations in $\Bkld$ up to time $n$. We do not indicate the dimension $d$ in $\bkln$ since it is always clear from the context.

It is possible to weaken the iid assumption for our test if the densities $(f_n)_{n\in\mathbb{N}}$ have uniform marginals, as shown in the following proposition.
\begin{proposition} \label{prop:non_iid}
Assume that $(f_n)_{n\in\mathbb{N}}$ in \eqref{eq:likelihood_ratio_process} have uniform marginals. Then $(M_n)_{n\in\mathbb{N}}$ is a test martingale under $H_0$ if at least one of the sequences $(X_n)_{n\in\mathbb{N}}$ and $(Y_n)_{n\in\mathbb{N}}$ is iid.
\end{proposition}
In Section \ref{sec:implementation}, we present a method to modify the densities $(f_n^{(d)})_{n\in\mathbb{N}}$ to achieve uniform marginals, which also turns out to have positive effects on the power of the test.

\begin{remark}[Filtration] \label{remark:filtration}
	The process $\mathcal{M}^{d}$ is a $(\mathcal{G}_n)_{n\in\mathbb{N}}$-martingale, where $\mathcal{G}_n = \sigma(R_i, S_i, i = 1, \dots, n)$, $n \in \mathbb{N}$, but not an $(\mathcal{F}_n)_{n\in\mathbb{N}}$-martingale. If $(X,Y)$ has continuous marginal distributions, then $R_n = \hat{F}_n(X_n) - U_n / n$ and so $U_n = R_n - \lfloor nR_n \rfloor /n$; analogously for $S_n$. Thus,
	\begin{equation} \label{eq:filtration}
		\mathcal{G}_n = \sigma\{\hat{F}_i(X_i
		), \hat{G}_i(Y_i), U_i, V_i,  i = 1, \dots, n\},
	\end{equation}
	and no information is lost when passing from the sequential ranks $\hat{F}_n(X_n), \hat{G}_n(Y_n)$ to the randomized ranks $(R_n, S_n)$. If the marginal distributions of $(X,Y)$ are not continuous, then \eqref{eq:filtration} does not hold. This difference is important if one wants to deramdomize the test martingales, c.f.~Section \ref{sec:implementation} and Section \ref{sec:derandomization}. Notice that, related to the discussion in Section \ref{sec:literature_review}, our test martingales $\mathcal{M}^{d}$ are safe --- i.e., satisfy \eqref{eq:stopped_martingale} --- only with respect to stopping rules under $(\mathcal{G}_n)_{n\in\mathbb{N}}$, but not under the original filtration $(\mathcal{F}_n)_{n\in\mathbb{N}}$. Such a restriction also appears in the group-invariant anytime-valid tests by \citet{Perez2022}, which are only safe with respect to the filtration generated by a maximal invariant.
\end{remark}

\begin{remark}[Connection to multinomial testing and prediction]
	The density $f_n^{(d)}$ equals the likelihood ratio between the observed bin frequencies, including one initial count per bin, and the uniform multinomial distribution with probabilities $d^{-2}$. So for fixed $d$, our test martingale $\mathcal{M}^d$ gives a test of the uniform multinomial distribution. Including an initial count of $1$ in each bin is equivalent to a uniform prior over the bin probabilities $\P\{(R_n, S_n) \in \Bkld\}$; that is,
	\begin{equation} \label{eq:prior}
		M_N^{(d)} = d^{2N} \int_{\Theta}\prod_{k,\ell=0}^{d-1} p_{k,\ell}^{\bkln} \, 
		\text{d}\pi\{(\pkl)_{k,\ell=0}^{d-1}\},
	\end{equation}
	where the integration is with respect to the uniform distribution $\pi$ over the unit simplex $\Theta$. Hence our method belongs to the general class of Bayes factor approaches for constructing test martingales \citep[Section 3.2.3]{Ramdas2022game}. An established measure for the power of a test martingale $(M_n)_{n\in\mathbb{N}}$ is the growth rate $\mathbb{E}_P\{\log(M_n)\}$ under an alternative hypothesis $P$ \citep{Grunwald2019, Shafer2021}. As discussed in \citet[Example 5]{Grunwald2019}, an asymptotically optimal choice for the prior in \eqref{eq:prior} in terms of growth rate would be Jeffreys' prior, i.e., the Dirichlet$(1/2,\dots,1/2)$ prior, and not the uniform distribution; see also \citet{Clarke1994, Xie2000}. We choose the uniform prior because it greatly simplifies the algebra in the derivation of bounds on type II error. If the bin probabilities $\P\{(R_n, S_n) \in \Bkld\}$ are bounded away from zero and one, which is typically the case if the dependence of $X$ and $Y$ is not very strong, then the loss in growth rate compared to Jeffreys' prior is only of constant order \citep[Chapter 8]{Gruenwald2007}. However, we emphasize that exactly the same calculations in the proofs would work with $1/2$, replacing factorials with the gamma function.
\end{remark}

\subsection{Power for fixed discretization} \label{sec:power}
We proceed to analyze the power of our test martingale under violations of the null hypothesis. Under independence, the randomized sequential ranks $(R_n, S_n)_{n\in\mathbb{N}}$ are a sequence of independent uniform random variables on $[0,1]^2$. The main challenge in deriving results on power is that when $X$ and $Y$ are dependent, the distribution of $(R_n, S_n)$ changes with $n$. As $n$ increases, convergence of the empirical CDFs $\hat{F}_n, \hat{G}_n$ implies that the distribution of $(R_n, S_n)$ becomes similar to that of the generalized probability integral transform
\begin{equation} \label{eq:honest_ranks}
	(\tilde{R}_n, \tilde{S}_n) = \{U_nF(X_n) + (1-U_n)F(X_n-), V_nG(Y_n)+(1-V_n)G(Y_n-)\}
\end{equation}
with the true marginals $F, G$ instead of the empirical CDFs $\hat{F}_n, \hat{G}_n$. Denote the distribution of $(\tilde{R}_n, \tilde{S}_n)$ by $Q$ in the following. Notice that the distribution of $(\tilde{R}_n, \tilde{S}_n)$ is indeed independent of $n$. Let $\Qd$ be the approximation of $Q$ on the $(d\times d)$-grid with Lebesgue density
\[
q^{(d)}(r,s) = d^2\prod_{k,\ell=0}^{d-1}\left(\qkld\right)^{\ONE\{(r,s) \in \Bkld\}}, \ \ (r,s) \in [0,1]^2,
\]
where $\qkld = \Q(\Bkld)$, $k, \ell = 0, \dots, d-1$. This density can be regarded as the population counterpart of the estimator $f_n^{(d)}$ in equation \eqref{eq:density_estimator}. All objects $Q$, $\Qd$ and $\qkld$ depend on the distribution $P$ of $(X,Y)$, which is, however, not indicated explicitly to lighten the notation.

In Theorem \ref{theorem:stopping} below, we show that the power of $\mathcal{M}^d$ depends on the Kullback-Leibler divergence $\kl(\Qd\|\mathcal{U})$. As a first step, we prove some useful results about $Q$ and $\Qd$.

\begin{proposition} \label{prop:population_properties}
	\begin{enumerate}
		\item[(i)] The distribution $Q$ is uniform if and only if $X$ and $Y$ are independent.
		
		\item[(ii)] If $\done$, $\dtwo$ are positive integers such that $\dtwo$ is a multiple of $\done$, then 
		\begin{align*}
			\kl(Q^{(d_1)}\|\mathcal{U}) \leq \kl(Q^{(d_2)}\|\mathcal{U}),
		\end{align*}
		with equality if and only if $Q^{(d_1)} = Q^{(d_d)}$.
		\item[(iii)] If $Q$ admits a density with respect to the Lebesgue measure on $[0,1]^2$, then
		\[
		\kl(\Qd\|\mathcal{U}) \rightarrow \kl(\Q\|\mathcal{U}), \ d \rightarrow \infty.
		\]
		\item[(iv)] If $Q$ is not absolutely continuous with respect to the Lebesgue measure on $[0,1]^2$, then
		\[
		\kl(\Qd\|\mathcal{U}) \rightarrow \infty, \ d \rightarrow \infty.
		\]
		
	\end{enumerate}
\end{proposition}
Proposition \ref{prop:population_properties} guarantees that $\kl(\Qd\|\mathcal{U}) > 0$ for $d$ large enough if $X$ and $Y$ are dependent. This result relies on established properties of the Kullback-Leibler divergence, and for part (iv), where $\kl(\Q\|\Ucal)$ is not defined, on the fact that $\Q$ has continuous marginal distributions.

\begin{remark}[Connection to mutual information]
	If $P$ is Lebesgue continuous with joint and marginal densities $p$ and $f,g$, respectively, then $Q$ is the distribution of $\{F(X),G(Y)\}$, and
	\[
	\kl(\Q\|\mathcal{U}) = \int_{\mathbb{R}^2} p(x,y)\log\left\{\frac{p(x,y)}{f(x)g(y)}\right\} \, d(x,y);
	\]
	see, e.g., \citet[Section 3]{Blumentritt2012}. Hence $\kl(\Qd\|\mathcal{U})$ is an approximation of the mutual information measure of dependence. This also holds for discrete $(X,Y)$, where the densities $p, f, g$ are replaced by probability mass functions and the integral by a sum over the support of $(X,Y)$. We are not aware of a simple characterization of $\kl(\Q\|\Ucal)$ in the case of mixed discrete-continuous distributions, but it directly follows from Proposition \ref{prop:population_properties} (i) that $X$ and $Y$ are dependent if and only if $\kl(\Q\|\Ucal) > 0$.
\end{remark}

Independence of $X$ and $Y$ can be rejected at the level $\alpha \in (0,1)$ as soon as the test martingale $\mathcal{M}^{d}$ exceeds $1/\alpha$. We define the corresponding stopping time
\[
\tau = \tau(\alpha, \mathcal{M}^{d}) = \inf(n \in \mathbb{N}\colon M_n^{(d)} \geq 1/\alpha),
\]
so that a rejection of independence before $N$ is equivalent to the event $\{\tau \leq N\}$. We first state our main technical result about the probability of this event, and then its implications for testing.

\begin{theorem}\label{theorem:stopping}
	If $N \geq 3$, $d \geq 2$, and $40d\log(N)^{1/2} \leq N^{1/2} e^{-1}$, then
	\begin{align*}
		\pr( \tau \leq N ) 
		& \geq \ONE\Big[N \kl(\Qd || \Ucal) - 47 d^{3} \{N\log(N)\}^{1/2} \geq \log(1/\alpha) \Big] 
		\times \\
		& \qquad\quad \Big[1 - 4\exp\Big\{ -127 d^{2} \log(N) \Big\}\Big].
	\end{align*}
	If, in addition, $\kl(\Q || \Ucal) < 1/d^{4}$, then
	\begin{align*}
		& \pr( \tau \leq N ) 
		\geq \mathbbm{1}\Big( \Big[ \{N \kl(\Qd || \Ucal)\}^{1/2} - \frac{40 d^{5} \log(N)^{1/2}}{1 - 2^{-1/2}} \Big] \{N \kl(\Qd || \Ucal)\}^{1/2} \geq \log(1/\alpha) \\
		& \qquad\qquad\qquad\qquad \ \ +3(d^2-1)\log(N)/2 + 5 d^{2}\Big)
		\times \Big[1 - 4\exp\Big\{ -127 d^{2} \log(N) \Big\}\Big].
	\end{align*}
\end{theorem}

As an immediate consequence of Theorem \ref{theorem:stopping}, we get the following corollary.

\begin{corollary}
	If $\kl(\Qd || \Ucal) > 0$, then $\mathbb{E}_P\{\tau(\alpha, \mathcal{M}^{d})\} < \infty$ for all $\alpha \in (0,1)$.
\end{corollary}

\begin{remark}
	Viewing our test in the classical fixed sample contiguity framework where we only collect $N$ samples, Theorem \ref{theorem:stopping} (ii) implies that the sum of type I and type II errors for testing
	\begin{align*}
		H_{0}:  \kl(\Qd || \Ucal) = 0, \quad
		H_{1}:  \kl(\Qd || \Ucal) = \frac{h_{d} \log(N)^2}{N}
	\end{align*}
	is strictly less than one if $h_{d} > 0$ is sufficiently large, depending only on $d$ but not on $N$.  The dependence on $N$ is optimal for testing multinomial distributions up to logarithmic factors. To see this, consider the binomial testing problem where we have $N$ iid Bernoulli variables with success probability $p$. Pinsker's inequality implies that
	\begin{align*}
		\Big(p - \frac{1}{2} \Big)^{2} \leq \frac{1}{2} \mathrm{KL}\Big\{\Ber(p) || \Ber(1/2)\Big\} \leq 4\Big(p - \frac{1}{2} \Big)^{2}
	\end{align*}
	if $1/4 \leq p \leq 3/4$.  Then, the detection boundary for this parametric problem is $H_0\colon p = 1/2$ versus $H_1\colon p = 1/2 + h/(N)^{1/2}$, implying that the detection boundary in terms of the Kullback-Leibler divergence is $H_0 \colon \mathrm{KL}\{\Ber(p) || \Ber(1/2)\} = 0$ versus $H_1\colon \mathrm{KL}\{\Ber(p) || \Ber(1/2)\} = h'/N$, for some value $h' > 0$ independent of $N$.  
	From our proofs, we note that $h_{d}$ can be bounded from above by $hd^{10}$ for some $h > 0$, though we believe this is an artifact of our proof technique and a more refined analysis can lower the dependence on $d$.
\end{remark}

In the literature on sequential tests for independence, the result most similar to our Theorem \ref{theorem:stopping} is by \citet[Theorem 2]{Shekhar2021}, which also yields uniform detection bounds. However, our uniform power guarantees are more limited in the sense that they require distributions whose deviation from independence can be detected on a the $(d\times d)$-discretization of $[0,1]^2$. Such a restriction is inevitable with a binning approach and also appears in Theorem 4.4 of \citet{Zhang2019}. \citet{Podkopaev2023} and \citet{Podkopaev2022} show that their tests have power $1$, but they do not derive uniform bounds on the type II error or prove finiteness of the stopping time for rejecting at level $\alpha$.

\subsection{Combining different discretizations} \label{sec:combination}
We have shown that for a fixed $d$, tests based on $\mathcal{M}^{d}$ have uniform power and finite expected stopping times if $\kl(\Qd\|\Ucal) > 0$, and now turn to the question of how to choose $d$. It follows from the proof of Theorem \ref{theorem:stopping} that $    M_N^{(d)} \geq \exp\{N\cdot \kl(\Qd\|\Ucal) + o_p(N)\}$,
so larger values of $d$ imply that $M_N^{(d)}$ eventually grows at a faster exponential rate. However, for high values of $d$, the estimators $f_n^{(d)}$ and the resulting test martingales $\mathcal{M}^{d}$ may perform poorly at small sample sizes.
To balance this trade-off, we propose to aggregate over different values of $d$. Such a mixture approach is standard in the construction of test martingales and confidence sequences.

There are two ways to aggregate over different discretization depths $d$. One strategy is to combine the density estimators $f_n^{(d)}$, and the other to combine the test martingales $\mathcal{M}^{d}$. For suitable weights, both strategies yield sequential tests with finite expected stopping time.

\begin{proposition}\label{proposition:aggregation}
	Assume that $P$ does not satisfy $H_0$. Let $(w_{d})_{d\in\mathbb{N}_0}$ be non-negative weights with $\sum_{d=1}^{\infty}w_d = 1$, such that $w_{d^*} > 0$ and $\kl(Q^{(\dstar)} || \Ucal) > 0$ for some $d^* \geq 1$. Define
	\begin{align*}
		& \mathcal{M}_0 = (M_{n,0})_{n\in\mathbb{N}}, \quad M_{N,0} = \prod_{n=1}^N \left\{w_0 + (1-w_0)\sum_{d=1}^{\infty} w_d f_{n}^{(d)}(R_n,S_n)\right\}, \ N \in \mathbb{N}, \\
		& \mathcal{M}_1 = (M_{n,1})_{n\in\mathbb{N}}, \quad M_{N,1} = \sum_{d=1}^{\infty}w_d M_n^{(d)}, \ N \in \mathbb{N}.
	\end{align*}
	Then $\mathbb{E}_P\{\tau(\alpha, \mathcal{M}_1)\} < \infty$ for all $\alpha \in (0,1)$. If $\sum_{d=1}^{\infty}w_dd^2 < \infty$, then $\mathbb{E}_P\{\tau(\alpha, \mathcal{M}_0)\} < \infty$.
\end{proposition}
In $\mathcal{M}_0$, there is the possibility to include a constant weight $w_0$. Such a correction is often applied the construction of test martingales to ensure that the mulitplicative increments are bounded away from zero; for example, \citet{Shekhar2021} use $w_0 = 0.5$ for their tests.

Interestingly, the two seemingly different combination approaches in Proposition \ref{proposition:aggregation} arise as special cases of established methods for the aggregation of density estimators, known under the names Gibb's estimator or mirror averaging, c.f.~\citet[Chapters 3 and 4]{Catoni2004} and \citet{Juditsky2008}. Reframed to our problem, these methods correspond to a test martingale
\begin{equation} \label{eq:mirror_averaging}
	M_{N,\eta} = \prod_{n=1}^N\left(\sum_{d=1}^{\infty}w_{d,n}f_n^{(d)}\right), \quad w_{d,n} = \frac{w_d(M_{n-1}^{(d)})^{\eta}}{\sum_{k=1}^{\infty}w_k(M_{n-1}^{(k)})^{\eta}},
\end{equation}
where $(w_d)_{d\in\mathbb{N}}$ are initial weights and $\eta \geq 0$ is a tuning parameter, called inverse temperature \citep{Catoni2004} or learning rate \citep{Turner2023}. For $\eta = 1$ one can easily show that \eqref{eq:mirror_averaging} is equivalent to the martingale aggregation, and for $\eta = w_0 = 0$ one obtains the density aggregation. Henceforth we refer to these methods via their value of $\eta$, i.e., $\mathcal{M}_{\eta}$ with $\eta = 0, 1$.

Due to Proposition \ref{prop:population_properties} (iii), the condition that $w_{d^*} > 0$ and $\kl(Q^{(\dstar)} || \Ucal) > 0$ for some $d^*$ is satisfied if for all $k \in \mathbb{N}$ there exists a $d \geq k$ with $w_d > 0$. In this case the computation of $\mathcal{M}_1$ and $\mathcal{M}_0$ is still feasible if one defines $f_n^{(d)} \equiv 1$ for all $n \leq n_d$, where $(n_d)_{d\in\mathbb{N}}$ is a strictly increasing sequence of integers. This does not affect the power guarantees of the test, and we have
\[
\sum_{d=1}^{\infty} w_d f_{n}^{(d)}(R_n,S_n) = \Big(1 - \sum_{d\colon n_d < n}w_d\Big) + \sum_{d\colon n_d < n}^{\infty} w_d f_{n}^{(d)}(R_n,S_n);
\]
the same decomposition can be done for $\sum_{d=1}^{\infty}w_d M_n^{(d)}$.

\subsection{Practical aspects and implementation} \label{sec:implementation}
In this section, we discuss finite sample corrections, implementation details, and practical aspects of our tests. Remarks on the efficient computation of sequential ranks are given in Section \ref{sec:efficient_implementation}.

In practice, one often wants to avoid that results that depend on external randomization, i.e., on $(U_n, V_n)_{n\in\mathbb{N}}$ in our case. If $(X,Y)$ has continuous marginals, a computationally efficient method is to derandomize $f_n^{(d)}$ at each time step. That is, with $\mathcal{H}_n = \sigma\{\hat{F}_i(X_i), \hat{G}_i(Y_i), i = 1, \dots, n\}$, we define the martingale $\bar{M}_N^{(d)} = \prod_{n=1}^N \mathbb{E}\{f_n^{(d)}(R_n, S_n)|\mathcal{H}_n\}$, where
\begin{align}
	& \mathbb{E}\{f_n^{(d)}(R_n, S_n) | \mathcal{H}_n\} = \sum_{k,\ell=0}^{d-1} \frac{\pr[\{\hat{F}_n(X_n) -U_n/n, \hat{G}_n(Y_n) -V_n/n\} \in \Bkld \mid \mathcal{H}_n] d^2h^{(d)}_{k\ell,n}}{n-1+d^2}, \nonumber \\
	& h^{(d)}_{k\ell,n} = 1+\sum_{m=1}^{n-1} \pr[\{\hat{F}_m(X_m) -U_m/m, \hat{G}_k(Y_m) -V_m/m\} \in \Bkld \mid \mathcal{H}_m]. \label{eq:derandomized_freq}
\end{align}
In the above equations, probabilities are only computed over $\Umat_n, \Vmat_n$, and $h_{k\ell,n}^{(d)}$ is the expected number of counts in $\Bkld$. The process $(\bar{M}_n)_{n \in \mathbb{N}}$ is still a martingale, since
\begin{align*}
	\mathbb{E}(\bar{M}_n^{(d)}|\mathcal{H}_{n-1}) & = \mathbb{E}[ \mathbb{E}\{f_n^{(d)}(R_n, S_n)\mid \mathcal{H}_n\} \bar{M}_{n-1}^{(d)}|\mathcal{H}_{n-1}] \\ 
    & =  \bar{M}_{n-1}^{(d)}\mathbb{E}\{f_n^{(d)}(R_n, S_n)| \mathcal{H}_{n-1}\} \\
    & = \bar{M}_{n-1}^{(d)},
\end{align*}
where the last equality holds because $f_n^{(d)}$ is a density, $(R_n, S_n)$ are uniform on $[0,1]^2$ and independent of $\mathcal{H}_{n-1}$. Derandomization for distributions with discontinuous marginals is discussed in Section \ref{sec:derandomization}.

Next, we discuss two finite sample corrections of our method. Due to Proposition \ref{prop:sequential_ranks} (ii), the marginal distributions of the sequential ranks $(R_n, S_n)$ are uniform, but this is not taken into account by the histogram estimator $f_n^{(d)}$ defined in \eqref{eq:density_estimator}. The estimator $f_n^{(d)}$ involves $d^2-1$ parameters, the bin frequencies, which could be reduced to $(d-1)^2$ when taking into account restrictions due to uniform marginals. Since the number of parameters is still of order $d^2$, one cannot expect a substantial improvement of the asymptotic properties of the test, but the finite-sample performance can be improved with suitable corrections. Our proposal is to apply Sinkhorn's algorithm \citep{Sinkhorn1964}. For this method, one arranges the bin frequencies, including the initial count of $1$ per bin, in a $d\times d$ matrix,
\[
C_n = \left(\frac{b_{k\ell,n-1} + 1}{n-1+d^2}\right)_{k,\ell=0}^{d-1}
\]
and alternately normalizes the rows and columns of $C_n$ to sum to $1/d$. Since the entries of $C_n$ are positive, this procedure converges to a matrix $\check{C}_n = (\check{c}_{k\ell})_{k,\ell=0}^{d-1}$ with row and column sums $1/d$. One can view $C_n$ as a probability distribution for a $d\times d$ contingency table, and as shown by \citet{Ireland1968} and \citet{Csiszar1975}, Sinkhorn's algorithm gives the information projection of $C_n$ on the space of distributions for contingency tables with uniform marginals. Define the density $\check{f}_n^{(d)}(r,s) = d^2\check{c}_{k\ell,n}$ for $(r,s)\in\Bkld$, $k,\ell=0,\dots,d-1$. It then follows from Theorem 2.2 of \citet{Csiszar1975} that
\[
\mathbb{E}_{P}\left[\log\left\{\check{f}_n^{(d)}(R_n, S_n)\right\} \middle| \mathcal{G}_{n-1}\right] \geq     \mathbb{E}_{P}\left[\log\left\{f_n^{(d)}(R_n, S_n)\right\} \middle| \mathcal{G}_{n-1}\right]
\]
under any distribution $P$ on $[0,1]^2$ for which $(R_n,S_n)$ have uniform marginals. Since this includes the data generating distribution, the corrected estimator is better in expected growth rate than $f_n^{(d)}$ in the theoretical limit of infinitely many iterations. The same strategy can also be applied to correct the derandomized counts \eqref{eq:derandomized_freq}, and the above result on the growth rate continues to hold. In our implementation, we run Sinkhorn's algorithm with a maximum of $20$ iterations, and stop early if all row and column sums are in $(1.001^{-1}/d, 1.001/d)$.

As for the second correction, since $n$th randomized ranks $(R_n, S_n)$ are randomized over a region of size $(1/n)\times(1/n)$, one cannot expect to gain a lot of information or a good estimator $f_n^{(d)}$ for very small sample sizes; for continuous data, we therefore set $f_n^{(d)} \equiv 1$ for $n \leq d$, and we replace $R_1, \dots, R_d$ by $\hat{F}_d(X_1), \dots, \hat{F}_d(X_d)$ in the computation of $f_n^{(d)}$ for $n > d$; analogously for $S_n$. This is valid since any function of $\hat{F}_1(X_1), \dots, \hat{F}_n(X_n)$ can be used to construct $f_{n+1}^{(d)}$, and the usual non-sequential ranks on the first $d$ observations, $\hat{F}_d(X_1), \dots, \hat{F}_d(X_d)$, are a function of the sequential ranks if the observations follow a continuous distribution.

Finally, we discuss the choice of the weights for aggregation. Motivated by the binary expansion testing (BET) method by \citet{Zhang2019}, we use test martingales with $d = 2^1, \dots, 2^{K}$ for some $K \in \mathbb{N}$ in our applications. This choice of $d$ tests whether the binary number expansions of $R_n, S_n$, $n\in\mathbb{N}$, are independent up to the $K$th binary digit. Following \citet{Zhang2019}, we consider a maximal depth of $K = 4$, which is sufficient for many practical applications even with weak dependence. For $\mathcal{M}_1$, we choose equal weights of $1/4$, and for $\mathcal{M}_0$ we set $w_0 = 0.2$, $w_{d} = 0.25$, $d = 2, 4, 8, 16$, which yields weight $0.2$ for each discretization depth. By our theory, these methods have guaranteed power to detect any dependence of $(X,Y)$ that yields non-uniform frequencies on the discretization of $[0,1]^2$ into the $256$ regular bins of size $(1/16)\times(1/16)$.

\section{Empirical results} \label{sec:simulations}

\subsection{Overview}
We assess the finite-sample performance of our tests in simulation examples by \citet{Zhang2019}, described in Table \ref{tab:simulations} and illustrated in Figure \ref{fig:sim_illustration_main}. Additional simulations and figures for the illustration of our methods are in Section \ref{sec:moresim}. All results, such as rejection rates, are computed over $1000$ independent replications of the simulations.
Computations were performed in \textsf{R} 4.1.3 and Python 3.8.5. Replication materials are available on \url{https://github.com/AlexanderHenzi/sequential_independence}. 

\begin{figure}[t]
	\centering
	\includegraphics[width=\textwidth]{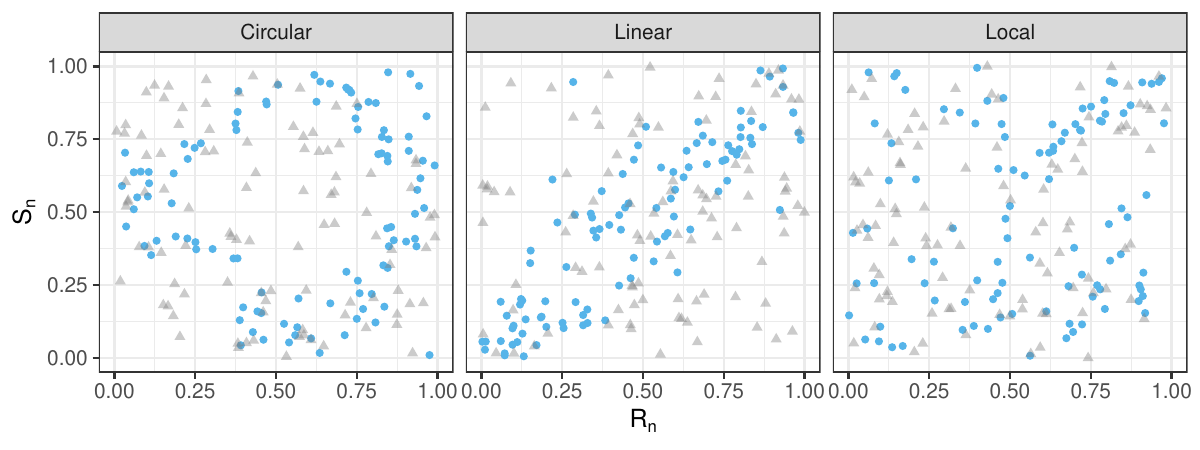}
	\caption{Randomized sequential ranks for the examples from Table \ref{tab:simulations} with $l = 1, 9$ (dots, triangles). \label{fig:sim_illustration_main}}
\end{figure}

\begin{table}[b]
	\centering
        \small
	\begin{tabular}{l l l}
		Scenario        &   Generation of $X$   &   Generation of $Y$   \\[0.25em]
		Circular        & $X = \cos(\theta)+2.5\epsilon$    & $Y = \sin(\theta) + 2.5\epsilon'$ \\
		Linear          & $X = U$               & $Y = X + 6\epsilon$   \\
		Local           & $X = G_1$             & $Y = \ONE(G_1, G_2 \in [0,1])(X+\epsilon) + [1-\ONE\{G_1, G_2 \in [0,1]\}]G_2$  \\[0.5em]
	\end{tabular}
	\caption{Simulation examples, with the following variables: for $l = 1, 3, 5, 7, 9$, $\epsilon, \epsilon' \sim \mathcal{N}\{0,(l/40)^2\}$; $\theta \sim \mathrm{Unif}[-\pi,\pi]$; $U \sim \mathrm{Unif}[0,1]$; $G_1, G_2 \sim \mathcal{N}(0,1/4)$.    \label{tab:simulations}}
\end{table}

\subsection{Methods}
The implementation of our tests is with the corrections described in Section \ref{sec:implementation}, and since the simulations only involve continuous distributions, all test martingales are derandomized. We compare against two competing methods. We implemented the sequential independence test by \citet[Section 5.2.2]{Shekhar2021} based on the Kolmogorov-Smirnov distance, which applies a martingale of the type \eqref{eq:betting_martingale} with functions $g_n$ of the form $g_n(x,y) = \ONE(x \leq u_n, y \leq v_n)$, where $u_n,v_n$ are chosen to maximize growth rate on the past observations. We search the optimal $u_n, v_n$ on a grid of spacing $0.025$. In the limit of an infinitely fine grid, this test has power against any alternative and is, like our test, invariant under strictly increasing transformations of $X$ and $Y$. As a second test, we apply the sequential kernelized independence test by \citet{Podkopaev2023}, which also constructs a martingale of the type \eqref{eq:betting_martingale} but with $g_n$ chosen as the witness function of the Hilbert-Schmidt independence criterion. We used code from \url{https://github.com/a-podkopaev/Sequential-Kernelized-Independence-Testing} and applied the test with the ONS and aGRAPA method for the parameters $\lambda_n$ and truncations of the betting functions at level $0.5$ and $0.9$, respectively. The bandwidths for the kernel are $1$ for the first $20$ observations, and estimated with the median heuristic on the first $20$ observations for the remaining part of the data. The test by \citet{Podkopaev2022} is not invariant under strictly increasing transformations of $X$ and $Y$. However, it remains valid under violations of the iid assumption for both variables, which stronger than our Proposition \ref{prop:non_iid}, and directly extends to multivariate data, which is more difficult for our method and a topic for future research.

The methods are abbreviated as ``Simple'' for our rank test without correction for uniform marginals, ``Sinkhorn'' for the corrected version with Sinkhorn's algorithm, both with $\eta = 0$ and $\eta = 1$ for the two combination approaches from Section \ref{sec:combination}; ``SR'' for the test by \citet{Shekhar2021}; and ``PR (ONS)'', ``PR (aGRAPA)'' for the test by \citet{Podkopaev2022}.

\subsection{Rejection rates}

Figure \ref{fig:rejection_rates} compares the rejection rates of the tests at the level $\alpha = 0.05$. For varying noise levels $l$, we display $N \mapsto P(\tau_{\alpha} \leq N)$, i.e., the distribution function of the rejection times. The method that performs best across all examples is our test with Sinkhorn's correction for uniform marginals. The aggregation of densities, $\eta = 0$, is generally preferable to the aggregation of martingales, $\eta = 1$, except for the regime with very low noise, and the Simple version of the test is outperformed by the Sinkhorn variant in all examples.
The method by \citet{Shekhar2021} generally has less power than our tests and performs best in the Linear simulation example. It seems that test functions of the form $g_n(x,y) = \ONE(x \leq u_n, y \leq v_n)$ are less suitable to detect the dependence at small $N$ in the given simulation examples. There is no uniform ranking between our methods and the test by \citet{Podkopaev2022}. While our tests detect dependence for small $l$ earlier, the test by \citet{Podkopaev2022} has better power in the Circular example for $l = 9$, where power $1$ is reached at about $4500$ observations, compared to $6500$ observations for the Sinkhorn variant with $\eta = 0$. Also in further simulations in Section \ref{sec:moresim}, neither of the methods uniformly dominates the other.

\begin{figure}[t]
	\centering
	\includegraphics[width=\textwidth]{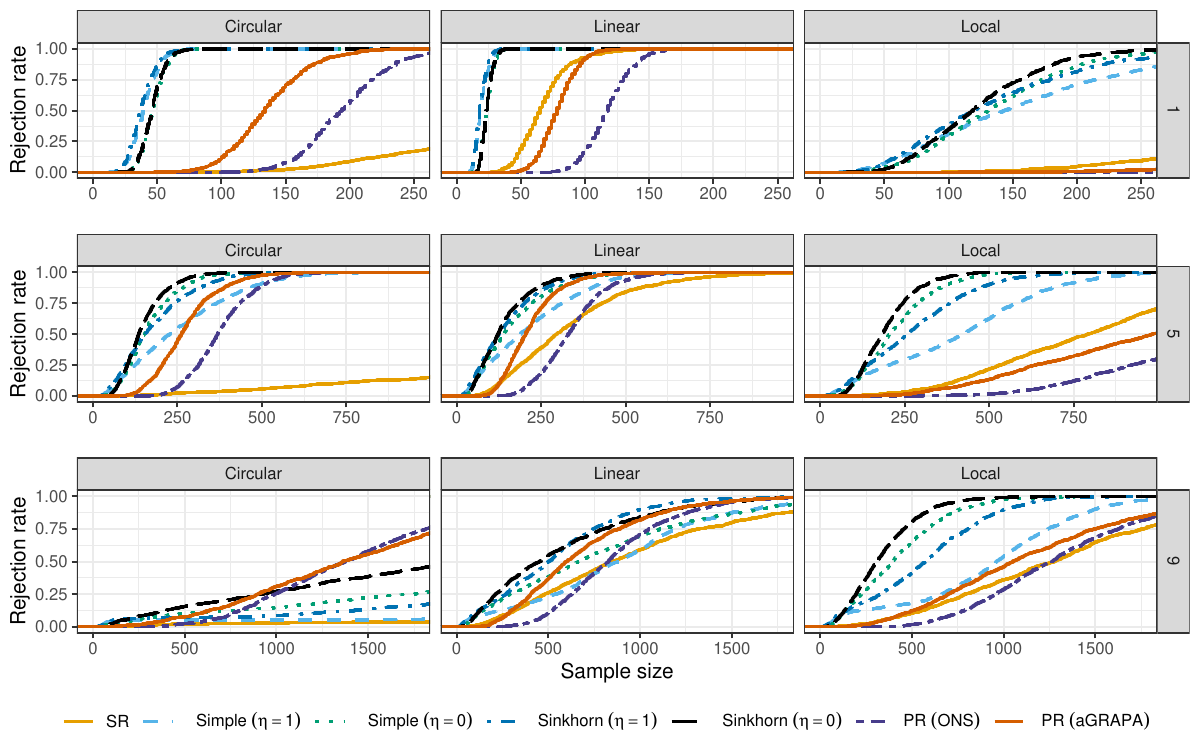}
	\caption{Rejection rates of the different methods for noise levels $l = 1, 5, 9$ (panel rows). \label{fig:rejection_rates}}
\end{figure}

\subsection{The gap in Ville's inequality, and truncated sequential tests} \label{sec:ville_gap}
Most of the recently developed tests for safe, anytime-valid inference rely on Ville's inequality \eqref{eq:ville} to achieve type I error control. Since the tests in this article are based on ranks, our test martingales are distribution free, which allows us to examine the gap in Ville's inequality numerically for our test. Moreover, if one is interested in collecting at most $N$ samples, one can increase the power of the test by approximating the minimal threshold $L_{\alpha, N}$ such that $P(M_n \geq L_{\alpha,N} \text{ for some } n \leq N) \leq \alpha$, for any of our variants for constructing $(M_n)_{n\in\mathbb{N}}$. This yields a sequential test valid up to $N$, and hence a pragmatic middle ground in between sequential tests with validity under infinite continuation and batch tests only valid for a fixed $N$.

We illustrate the above with the Sinkhorn variant of our test with $\eta = 0$. We estimate the CDF of $p_N = (\max_{n=1,\dots,N}M_n)^{-1}$ for different $N$ over $100\,000$ simulations with independent $(X,Y)$; a plot is in Section \ref{sec:moresim}. With $N = 4096$ observations, the probability $P(p_{4096} \leq 0.05)$ is approximately $0.0460$, and the calibrated rejection threshold is $L_{0.05,4096} = 18.3$ instead of $20$, so there is little advantage to performing the adjustment. However, for small $N$ the threshold can of course be substantially lowered, e.g.~to $9.17, 13.8, 16.9$ for $N = 128, 256, 512$, which is beneficial if one is only interested in collecting at most $N$ samples.

\begin{table}[t]
	\centering
	\small
	\begin{tabular}{llllllllcc}
		Simulation $\quad$ & $l \quad$ & \multicolumn{2}{l}{$\kl(Q^{(d)}\|\Ucal)\quad$} & \multicolumn{4}{l}{BET with $n$ samples $\qquad\qquad$} & \multicolumn{2}{c}{Sequential rank test} \\[0.5em]
	& & 2	& 8 & 64 & 128 & 256 & 512 & Power & Mean sample size\\[0.25em]

 & 1 & 0.00 & 0.72 & 1.00 & 1.00 & 1.00 & 1.00 & 1.00 & 45\\
 & 3 & 0.00 & 0.28 & 0.88 & 1.00 & 1.00 & 1.00 & 1.00 & 72\\
Circular & 5 & 0.00 & 0.10 & 0.28 & 0.72 & 1.00 & 1.00 & 1.00 & 144\\
 & 7 & 0.00 & 0.03 & 0.06 & 0.15 & 0.59 & 0.94 & 0.68 & 342\\
 & 9 & 0.00 & 0.01 & 0.02 & 0.04 & 0.13 & 0.29 & 0.17 & 467\\[0.1em]
 & 1 & 0.33 & 0.69 & 1.00 & 1.00 & 1.00 & 1.00 & 1.00 & 23\\
 & 3 & 0.08 & 0.16 & 0.75 & 0.99 & 1.00 & 1.00 & 1.00 & 59\\
Linear & 5 & 0.03 & 0.06 & 0.25 & 0.63 & 0.96 & 1.00 & 1.00 & 135\\
 & 7 & 0.02 & 0.03 & 0.12 & 0.33 & 0.66 & 0.97 & 0.86 & 252\\
 & 9 & 0.01 & 0.02 & 0.07 & 0.18 & 0.44 & 0.82 & 0.58 & 357\\[0.1em]
 & 1 & 0.00 & 0.21 & 0.12 & 0.44 & 0.95 & 1.00 & 1.00 & 121\\
 & 3 & 0.00 & 0.13 & 0.08 & 0.28 & 0.88 & 1.00 & 1.00 & 144\\
Local & 5 & 0.00 & 0.09 & 0.06 & 0.17 & 0.68 & 0.99 & 1.00 & 189\\
 & 7 & 0.00 & 0.06 & 0.05 & 0.13 & 0.45 & 0.85 & 0.95 & 248\\
 & 9 & 0.00 & 0.05 & 0.06 & 0.10 & 0.31 & 0.64 & 0.83 & 305\\
	\end{tabular}
	\caption{Kullback-Leibler divergence $\kl(Q^{(d)}\|\Ucal)$, $d = 2, 8$, for the simulation examples from Table \ref{tab:simulations} with $l = 1, 3, 5,7, 9$; power of the BET with $n = 64, 128, 256, 512$; and power and mean sample size of our sequential test truncated at $N = 512$. \label{tab:seq_vs_nonseq}}
\end{table}

The truncated version of our sequential test allows for a meaningful comparison with non-sequential methods. Consider a situation where a researcher has the budget to collect $512$ samples to test independence, but tries to minimize the sample size while still aiming for high power. With a non-sequential test, the researcher has to make an assumption about the strength of dependence and choose the sample size accordingly. Our truncated sequential test, on the other hand, automatically adapts to the strength of the dependence, and one only has to fix the upper limit on the sample size. In Table \ref{tab:seq_vs_nonseq} we compare the test by \citet{Zhang2019} with sample sizes $64$, $128$, $256$, and $512$ to our sequential test truncated at $512$, i.e., with rejection threshold $16.9$. Under the null our test rejects with probability $0.047$, and, hence, the simulated threshold controls the error probability. With maximum sample size $n = 512$ and for $l=7,9$, the test by \citeauthor{Zhang2019} has more power than ours, except for the Local simulation example. However, balancing power and sample size is a difficult task. For instance, $n = 256$ would be a good choice in the Linear example if one expects that $l = 5$, since it gives a power of $0.96$. But this sample size is unnecessarily large if $l = 1, 2$, and yields insufficient power for $l = 7,9$. Our sequential test adapts and rejects with power $1$ and $135$ observations, on average, for $l = 5$, and with only $23$ or $59$ observations for $l = 1$ and $l = 3$, respectively. For $l = 7,9$, the power is still higher than choosing $n = 256$ for \citeauthor{Zhang2019}'s test, but here it would be better to apply the latter with $n = 512$, since it achieves the highest power. Variations of this simulation example and comparisons against other tests, which yield the same conclusions, are presented in Section \ref{sec:moresim}.

To summarize, if one has prior knowledge about the strength of dependence and required sample sizes, then applying a well-designed non-sequential test can be more efficient than our methods. Otherwise, it is more safe to apply a sequential test, which automatically rejects early under strong dependence and is not underpowered if the dependence is weaker than expected.

\subsection{Illustration on real data}

\begin{figure}[t]
	\centering
	\includegraphics[width=\textwidth]{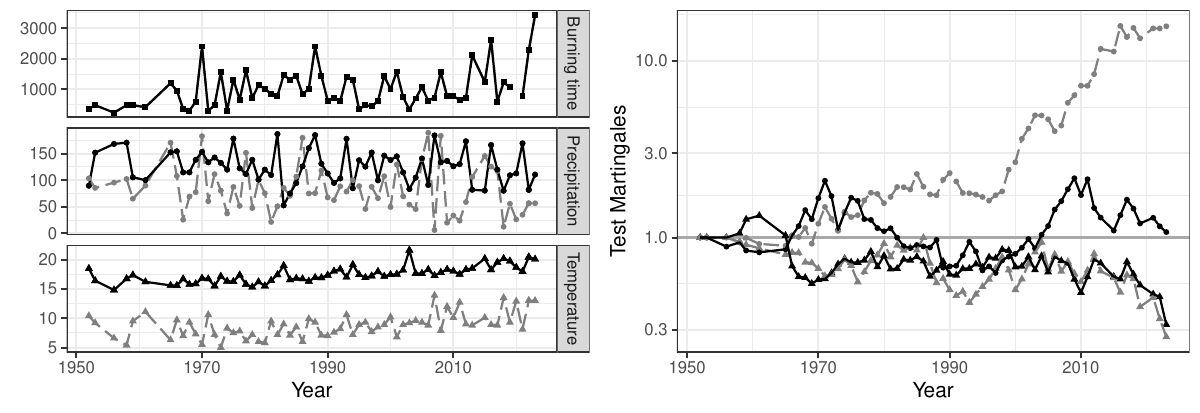}
	\caption{Burning time of the B\"o\"ogg (seconds), average precipitation (millimeters) and average temperature (degrees celsius) in the months June to August (black) and in April (gray); and test martingales for assessing dependence of temperature (lines with triangles) and precipitation (dots) and the burning time. \label{fig:sechselaeuten}}
\end{figure}

S\"achsil\"u\"ute is a traditional spring holiday in April in Z\"urich that celebrates the end of winter working hours and the transition to summer.  The climax of the holiday is the burning of the B\"o\"ogg, a snowman-like figure that is filled with explosives.  Similar to Groundhog Day in the United States, local folklore says that a long burning time until the explosion of the B\"o\"ogg announces a cold and wet summer, while a shorter burning time implies better weather in the upcoming summer.  Despite any scientific foundation, the predictions of the B\"o\"ogg receive a lot of attention in the Swiss media, particularly the record-breaking burning time of 57 minutes in 2023. From a scientific point of view, a natural explanatory factor for the different burning times is the weather at the S\"achsil\"u\"ute, with higher precipitation implying longer burning times.

For regularly recurring, yet infrequent, observations, like the burning time of the B\"o\"ogg and the associated weather, it is natural to apply a sequential test to monitor the outcomes continuously over time. Indeed, in between the first version and the revision of this article the data for 2023 became available, with which we could update the test martingales without any need for corrections for multiple testing. This would not be possible with non-sequential tests.

We applied our test to assess whether there is a dependence between either average temperature or precipitation in the months June to August and in April and the burning time of the B\"o\"ogg, with the April weather serving as a proxy for the weather at the holiday. The data for the burning times is available at \url{https://github.com/philshem/Sechselaeuten-data}, and the weather data is from \url{https://www.meteoswiss.admin.ch}. For our test, we consider the grids with $d = 2, 4$ and we apply the Sinkhorn correction for uniform marginals and $\eta = 0$.  There are a small number of ties in the data, and so we do not derandomize the martingales. As in Section \ref{sec:ville_gap}, we also simulated a rejection threshold for a truncated version of this test, with $\alpha = 0.05$ and a maximal sample size of $N = 128$, which allows continuing the sequential test far into the future. The simulated threshold equals $8.51$. We remark that the temperatures have an increasing trend and, hence, are not iid.  Since one can safely assume that the burning times are iid, our test is still valid due to Proposition \ref{prop:non_iid}.

The test martingales in Figure \ref{fig:sechselaeuten} provide no evidence against independence for the burning times and the summer weather, with final values of $0.32$ and $1.07$ for temperature and precipitation, respectively. The martingale for precipitation in April, on the other hand, reaches $15.71$, which allows to reject independence at the level $\alpha = 0.05$ with the truncated version of the test. While such a test result is unlikely to attract as much attention as the B\"o\"ogg's weather prediction, a formulation in terms of betting, as proposed by \citet{Shafer2021}, might be more intuitive for a broad audience. Namely, someone who believed the folklore and bet an equal amount of money on dependence of the burning time on precipitation and temperature in summer would have lost $31.5\%$ of the investment by $2023$, whereas the scientist who bet on dependence on precipitation in April would have multiplied the investment by a factor of almost $16$, corresponding to a yearly return of about $4\%$.

\section*{Acknowledgement}
We are grateful to Aleksandr Podkopaev and two anonymous referees for helpful comments.  AH is supported in part by European Research Council (ERC) under European Union's Horizon 2020 research and innovation programme, grant agreement No. 786461.  ML is supported in part by NSF Grant DMS-2203012.

\setlength{\bibsep}{0pt plus 0.3ex}
\bibliography{references}
\bibliographystyle{chicago}

\appendix

\section{Proofs} \label{sec:proofs}

\subsection{Proof of Theorem 1}
\begin{proof}
\let\boldsymbol\mathsf
    Since $M_{n}$ is an $\mathcal{F}_{n}$-supermartingale by assumption, we have that $M_{n}$ is $\mathcal{F}_{n}$-adapted.  Hence, by the Doob-Dynkin lemma, there exists functions $h_{n} : \mathbb{R}^{2n} \to \mathbb{R}$ and $h_{n+1} : \mathbb{R}^{2n+2} \to \mathbb{R}$ such that $M_{n} = h_{n}(\boldsymbol{X}_{n}, \boldsymbol{Y}_{n})$ and $M_{n+1} = h_{n}(\boldsymbol{X}_{n+1}, \boldsymbol{Y}_{n+1})$.  

    For part (i), it suffices to show that for any $\boldsymbol{x}_{n+1} = (\boldsymbol{x}_{n}, x_{n+1}), \boldsymbol{y}_{n+1} = (\boldsymbol{y}_{n}, y_{n+1}) \in \mathbb{R}^{n+1}$, we have $h_{n}(\boldsymbol{x}_{n}, \boldsymbol{y}_{n}) \geq h_{n+1}(\boldsymbol{x}_{n+1}, \boldsymbol{y}_{n+1})$.  To this end, fix arbitrary vectors $\boldsymbol{x}_{n+1}, \boldsymbol{y}_{n+1} \in \mathbb{R}^{n+1}$.  Let $\tilde{x}_{1}, \dots, \tilde{x}_{n_{x}}$ denote the distinct values of $\{x_{1}, \dots, x_{n}\} \backslash \{x_{n+1}\}$ and $\tilde{x}_{n_{x}+1} = x_{n+1}$ and similarly for $\tilde{y}_{1}, \dots, \tilde{y}_{n_{y}+1}$.  For $\epsilon > 0$, define the probability measure $P_{\epsilon}$ to be the product measure $P_{\epsilon} = P_{\epsilon, X} \times P_{\epsilon, Y}$ on $\mathbb{R} \times \mathbb{R}$, where $P_{\epsilon, X}(X = \tilde{x}_{i}) = \epsilon / n_{x}$ for all $i = 1, \dots, n_{x}$ and $P_{\epsilon, X}(X = x_{n+1}) = 1 - \epsilon$ and similarly for $P_{\epsilon, Y}$.  Note that $P_{\epsilon}$ satisfies $H_{0}$ for all $\epsilon > 0$.  Now, write $A = \{X_{1} = x_{1}, \dots, X_{n} = x_{n}, Y_{1} = y_{1}, \dots, Y_{n} = y_{n}) \in \mathcal{F}_{n}$.  Since $M_{n}$ is an $\mathcal{F}_{n}$-supermartingale relative to the measure $P_{\epsilon}$, it follows that 
    \begin{align*}
        P_{\epsilon}(A) h_{n}(\mathsf{x}_{n}, \mathsf{y}_{n}) 
        &= \mathbb{E}_{P_{\epsilon}} \{\ONE(A) M_{n}\} \\ 
        &\geq \mathbb{E}_{P_{\epsilon}} \{\ONE(A) \mathbb{E}_{P_{\epsilon}}( M_{n+1} | \mathcal{F}_{n})\} \\
        &= P_{\epsilon}(A) \mathbb{E}_{P_{\epsilon}}\{ h_{n+1}(\mathsf{x}_{n}, X_{n+1}, \mathsf{y}_{n}, Y_{n+1})\}.
    \end{align*}
    By construction, we have $P_{\epsilon}(A) > 0$, implying that 
    \begin{align*}
         h_{n}(\mathsf{x}_{n}, \mathsf{y}_{n}) \geq \mathbb{E}_{P_{\epsilon}}\{h_{n+1}(\mathsf{x}_{n}, X_{n+1}, \mathsf{y}_{n}, Y_{n+1})\}.
    \end{align*}
    As $\epsilon \to 0$, the right hand side converges to $h_{n+1}(\mathsf{x}_{n+1}, \mathsf{y}_{n+1})$, and combining the above calculations yields
    \begin{align*}
        h_{n}(\boldsymbol{x}_{n}, \boldsymbol{y}_{n}) \geq h_{n+1}(\boldsymbol{x}_{n+1}, \boldsymbol{y}_{n+1})
    \end{align*}
    for all $\boldsymbol{x}_{n+1}, \boldsymbol{y}_{n+1} \in \mathbb{R}^{2n+2}$, implying that $M_{n} \geq M_{n+1}$ everywhere.

    For the second part, fix $(\mathsf{x}_{n+1}, \mathsf{y}_{n+1}) = (\mathsf{x}_{n}, x_{n+1}, \mathsf{y}_{n}, y_{n+1}) \in \mathbb{R}^{2n+2}$ such that $x_{1}, \dots, x_{n+1}$ are pairwise distinct and similarly for $y_{1}, \dots, y_{n+1}$.  Furthermore, we assume that $(\mathsf{x}_{n}, \mathsf{y}_{n})$ is a Lebesgue point of $h_{n}$ and $(\mathsf{x}_{n}, x_{i}, \mathsf{y}_{n}, y_{j})$ is a Lebesgue point of $h_{n+1}$ for all $i,j = 1, \dots, n+1$.  For simplicity, we denote this subset of $\mathbb{R}^{2n+2}$ by $\mathcal{V}_{n}$.  Since vectors with non-distinct entries are contained in a finite union of proper linear subspaces and almost every point is a Lebesgue point, the complement of $\mathcal{V}_{n}$ has Lebesgue measure zero.  Next, letting
    \begin{align*}
        \epsilon \in (0, \min\{|x_{i} - x_{j}| + |y_{i} - y_{j}|:i,j = 1, \dots, n+1, i \neq j\} / 4),
    \end{align*}
    define the density of $X$ by
    \begin{align*}
        \frac{dP_{\epsilon, X}}{d\Leb}(x) = \frac{1}{2n} \sum_{i=1}^{n} 1_{|x - x_{i}| < \epsilon} + \frac{1-\epsilon}{2\epsilon} 1_{|x - x_{n+1}| < \epsilon}
    \end{align*}
    and analogously for $Y$.  We set $P_{\epsilon}$ to be the product distribution on $\mathbb{R}\times \mathbb{R}$ with these marginals.  Finally, let 
    \begin{align*}
    A_{\epsilon} & = [x_1 - \epsilon, x_1 + \epsilon] \times \dots \times [x_n - \epsilon, x_n + \epsilon] \times [y_1 - \epsilon, y_1 + \epsilon] \times \dots \times [y_n - \epsilon, y_n + \epsilon], \\
    B_{i,j,\epsilon} & = [x_i - \epsilon, x_i + \epsilon] \times [y_j - \epsilon, y_j + \epsilon] \ \subset \ \mathbb{R}^2, \ i,j = 1, \dots, n +1.
    \end{align*}
    By assumption, $M_{n}$ is an $\mathcal{F}_{n}$-supermartingale, implying
    \begin{align*}
        \mathbb{E}_{P_{\epsilon}} [\ONE\{(\mathsf{X}_{n}, \mathsf{Y}_{n}) \in A_{\epsilon}\} M_{n}] 
        \geq \mathbb{E}_{P_{\epsilon}} [\ONE\{(\mathsf{X}_{n}, \mathsf{Y}_{n}) \in A_{\epsilon}\} M_{n+1}].
    \end{align*}
    Next, we multiply both sides of the above display by $(n / \epsilon)^{2n}$ and compute the limit as $\epsilon \to 0$ separately.  Since 
    \begin{align*}
        [x_i - \epsilon, x_i + \epsilon] \cap [x_j - \epsilon, x_j + \epsilon] =
        [y_i - \epsilon, y_i + \epsilon] \cap [y_j - \epsilon, y_j + \epsilon] 
        = \varnothing
    \end{align*}
    for $i \neq j$ by the definition of $\epsilon$, it follows that
    \begin{align*}
        \Big(\frac{n}{\epsilon}\Big)^{2n} \mathbb{E}_{P_{\epsilon}} [\ONE\{(\mathsf{X}_{n}, \mathsf{Y}_{n}) \in A_{\epsilon}\} M_{n}] 
        &= \Big(\frac{n}{\epsilon}\Big)^{2n}\frac{1}{(2n)^{2n}} \int_{A_{\epsilon}} h_{n}(\mathsf{s}_{n}, \mathsf{t}_{n}) \Leb(d\mathsf{s}_{n}, d\mathsf{t}_{n}) \\
        &= \frac{1}{(2\epsilon)^{2n}} \int_{A_{\epsilon}} h_{n}(\mathsf{s}_{n}, \mathsf{t}_{n}) \Leb(d\mathsf{s}_{n}, d\mathsf{t}_{n}) \\
        &= \frac{1}{\Leb(A_{\epsilon})} \int_{A_{\epsilon}} h_{n}(\mathsf{s}_{n}, \mathsf{t}_{n}) \Leb(d\mathsf{s}_{n}, d\mathsf{t}_{n}).
    \end{align*}
    Since $(\mathsf{x}_{n}, \mathsf{y}_{n})$ is a Lebesgue point of $h_{n}$, as $\epsilon \to 0$, the Lebesgue differentiation theorem implies that
    \begin{align*}
        \lim_{\epsilon \to 0} \Big(\frac{n}{\epsilon}\Big)^{2n} \mathbb{E}_{P_{\epsilon}} [\ONE\{(\mathsf{X}_{n}, \mathsf{Y}_{n}) \in A_{\epsilon}\} M_{n}] 
        = h_{n}(\mathsf{x}_{n}, \mathsf{y}_n).
    \end{align*}
    Similarly, we have
    \begin{align*}
        &\Big(\frac{n}{\epsilon}\Big)^{2n} \mathbb{E}_{P_{\epsilon}} [\ONE\{(\mathsf{X}_{n}, \mathsf{Y}_{n}) \in A_{\epsilon}\} M_{n+1}] \\
        &\phantom{\leq}= \frac{\epsilon^{2}}{n^{2}}\frac{1}{(2\epsilon)^{2n+2}}
        \sum_{i,j=1}^{n} \int_{A_{\epsilon}} \Leb(d\mathsf{s}_{n}, d\mathsf{t}_{n}) \int_{B_{i,j, \epsilon}} h_{n+1}(\mathsf{s}_{n+1}, \mathsf{t}_{n+1}) \Leb(ds_{n+1}, dt_{n+1}) \\
        &\phantom{\leq\leq} + \frac{\epsilon(1-\epsilon)}{n} \frac{1}{(2\epsilon)^{2n+2}}
        \sum_{i=1}^{n} \int_{A_{\epsilon}} \Leb(d\mathsf{s}_{n}, d\mathsf{t}_{n}) \int_{B_{n+1,i, \epsilon}\cup B_{i,n+1, \epsilon}} h_{n+1}(\mathsf{s}_{n+1}, \mathsf{t}_{n+1}) \Leb(ds_{n+1}, dt_{n+1}) \\ 
        &\phantom{\leq\leq}+ (1-\epsilon)^{2}\frac{1}{(2\epsilon)^{2n+2}} \int_{A_{\epsilon}} \Leb(d\mathsf{s}_{n}, d\mathsf{t}_{n}) \int_{B_{n+1,n+1, \epsilon}} h_{n+1}(\mathsf{s}_{n+1}, \mathsf{t}_{n+1}) \Leb(ds_{n+1}, dt_{n+1}).
    \end{align*}
    Noting that $\Leb(A_{\epsilon} \times B_{i,j, \epsilon}) = (2\epsilon)^{-(2n+2)}$ for all $i, j = 1, \dots, n+1$, the Lebesgue differentiation theorem further implies that 
    \begin{align*}
        \lim_{\epsilon \to 0} \Big(\frac{n}{\epsilon}\Big)^{2n} \mathbb{E}_{P_{\epsilon}} [\ONE\{(\mathsf{X}_{n}, \mathsf{Y}_{n}) \in A_{\epsilon}\} M_{n+1}] 
        = h_{n+1}(\mathsf{x}_{n+1}, \mathsf{y}_{n+1}).
    \end{align*}
    Combining the above calculations shows that, for $(\mathsf{x}_{n+1}, \mathsf{y}_{n+1}) \in \mathcal{V}_{n}$, we have 
    \begin{align*}
        h_{n}(\mathsf{x}_{n}, \mathsf{y}_n) \geq h_{n+1}(\mathsf{x}_{n+1}, \mathsf{y}_{n+1}).
    \end{align*}
    Finally, letting $P'$ be an arbitrary distribution absolutely continuous with respect to Lebesgue measure, we conclude that
    \begin{align*}
        P'(M_{n} \geq M_{n+1}) = P'\{h_{n}(\mathsf{X}_{n}, \mathsf{Y}_{n}) \geq h_{n+1}(\mathsf{X}_{n+1}, \mathsf{Y}_{n+1}), (\mathsf{X}_{n+1}, \mathsf{Y}_{n+1}) \in \mathcal{V}_{n}\} = 1.
    \end{align*}
\end{proof}

\subsection{Proof of Proposition 1} \label{sec:discontinuous_marginals}
\begin{proof}
In the case of continuous marginals, part (i) of the Proposition is Theorem 1.1 by \citet{BarndorffNielsen1963}, and parts (ii) and (iii) are direct consequences of (i). It only remains to prove (ii) and (iii) for the case that the marginal CDFs are not continuous. We do this by giving an alternative definition of randomized sequential ranks, for which we prove that (i), (ii), and (iii) hold, and then show that the randomized sequential ranks \eqref{eq:randomized_ranks} are equivalent in distribution to the alternative definition given here.

Let $\delta_1, \dots, \delta_N$ be independent random variables, also independent of $X_1, \dots, X_N$, with a continuous distribution. Define
\begin{align} \label{eq:alternative_seq_ranks}
	\check{R}_n = \frac{1}{n}\sum_{i=1}^n \ONE(X_i < X_n) + \frac{1}{n}\sum_{i=1}^n\ONE(X_i = X_n, \delta_i \leq \delta_n), \ n = 1, \dots, N.
\end{align}
These are the sequential ranks of the bivariate observations $(X_n, \delta_n)$, $n = 1, \dots, N$, with the lexicographic ordering, i.e., $(X_i, \delta_i) \prec (X_j, \delta_j)$ if $X_i < X_j$ or if $X_i = X_j$ and $\delta_i \leq \delta_j$.

We first show that these sequential ranks satisfy part (i), with completely analogous arguments as in the proof of Theorem 1 of \citet{BarndorffNielsen1963}.
Since $\delta_1, \dots, \delta_N$ follow a continuous distribution, there are no pairs $i,j$ for which $(X_i, \delta_i) \prec (X_j, \delta_j)$ and $(X_j, \delta_j) \prec (X_i, \delta_i)$ both hold. Let $\xi_1, \dots, \xi_N \in \{1, \dots, N\}$ be the usual ranks of $(X_1, \delta_1), \dots, (X_N, \delta_N)$ in the order $\prec$. Because $(X_1, \delta_1), \dots, (X_N, \delta_N)$ are iid, we have
\[
    P(\xi_1 = k_1, \dots, \xi_N = k_N) = 1/N!
\]
for all permutations $(k_1, \dots, k_N)$ of $(1, \dots, N)$. Moreover, there is a bijection between $\xi_1, \dots, \xi_N$ and $\check{R}_1, \dots, \check{R}_N$, see the explicit formulas in \citet{Khmaladze1986}. Hence,
\[
    1/N! = P(\xi_1 = k_1, \dots, \xi_N = k_N) = P(\check{R}_1 = r_1, \dots, \check{R}_N = r_N)
\]
for all combinations of 
\[
    r_1 = 1, \ r_2 \in \{1/2, 1\}, \ r_3 \in \{1/3, 2/3, 1\}, \ \dots, \ r_N \in \{1/N, 2/N, \dots, 1\},
\]
which is the set where $\check{R}_1, \dots, \check{R}_N$ take their values. Consequently,
\[
    P(\check{R}_n = r_n) = \frac{1}{N!}\prod_{i = 1, \dots, N, \, i \neq n} \!\!\! i \ = \ \frac{1}{n}, \ r_n \in \{1/n \dots, 1\}, \ n = 1, \dots, N, 
\]
and so $P(\check{R_1} = r_1, \dots, \check{R}_N = r_N) = 1/N! = \prod_{n=1}^N P(\check{R}_n = r_N)$, which implies independence of $\check{R}_1, \dots, \check{R}_N$. Hence, $\check{R}_n - U_n / n$, $n = 1, \dots, N$, are iid uniform on $[0,1]$. The same arguments can be applied to $Y_1, \dots, Y_N$, and so (ii) and (iii) hold for the alternative definition \eqref{eq:alternative_seq_ranks} of randomized sequential ranks.

It remains to show that $R_1, \dots, R_N$ are equal in distribution to $\check{R}_n - U_n/n$, $n = 1, \dots, N$. Condition on $X_1, \dots, X_N$. Let $x \in \{X_1, \dots, X_N\}$, and let $n_1, \dots, n_k$ be the indices of the observations with $X_{n_j} = x$. Then,
\[
	D_{j} = \frac{1}{n_j}\sum_{i=1}^{n_j}\ONE(X_{i} = x, \delta_{i} \leq \delta_{j}) \in \{1/n_j, \dots, \hat{F}_{n_j}(x) - \hat{F}_{n_j}(x-)\}, \ j = 1, \dots, k,
\]
are, up to normalization, the sequential ranks of $\delta_{n_1}, \dots, \delta_{n_k}$, and so they are independent with
\[
	P\left(D_{j} = i/n_j\right) = \frac{1}{j}, \ i = 1, \dots, j, \ j = 1, \dots, k.
\]
Because $\check{R}_{n_j} = \hat{F}_{n_j}(x-) + D_{j}$,  $j = 1, \dots, k$, the above derivations show that $\check{R}_1, \dots, \check{R}_N$ are equal in distribution to
\[
	\hat{F}_n(X_n-) + B_n, \ B_n \sim \mathrm{Unif}\left[\left\{\frac{1}{n}, \frac{2}{n}, \dots, \hat{F}_n(X_n) - \hat{F}_n(X_n-)\right\}\right], \ n = 1, \dots, N,
\]
where $B_1, \dots, B_N$ are independent. Since $B_n - U_n/n$, $n = 1, \dots, N$, are independent and $(B_n-U_n/n)/\{\hat{F}_n(X_n) - \hat{F}_n(X_n-)\}$ is uniform on $(0,1)$, we obtain that $\check{R}_n - U_n/n$, $n = 1, \dots, N$, are equal in distribution to $\hat{F}_n(X_n-) + U_n\{\hat{F}_n(X_n) - \hat{F}_n(X_n-)\}$, $n = 1, \dots, N$, which completes the proof.
\end{proof}

\subsection{Proof of Proposition 2}
\begin{proof}
Assume that $(Y_n)_{n\in\mathbb{N}}$ is iid. Then,
\begin{align*}
    & \mathbb{E}\{f_n(R_n, S_n) \mid R_i, S_i, i = 1, \dots, n - 1\} \\
    & = \mathbb{E}[\mathbb{E}\{f_n(R_n, S_n) \mid R_n\} \mid R_i, S_i, i = 1, \dots, n - 1] \\
    & = \mathbb{E}\left\{\int_{0}^1 f_n(R_n, s) \, ds \mid R_i, S_i, i = 1, \dots, n - 1\right\} \\
    & = \mathbb{E}\left(1 \mid R_i, S_i, i = 1, \dots, n - 1\right) \\
    & = 1.
\end{align*}
The second equality above holds by independence of $S_n$ from $R_n, R_i, S_i$, $i = 1, \dots, n - 1$ under the null hypothesis and because $S_n$ is uniformly distributed, and the third holds because $r \mapsto \int_{0}^1 f_n(r, s) \, ds$ is the marginal of $f_n$, which is uniform on $[0,1]$ by assumption, so constant $1$. It follows that $\mathbb{E}(M_N | \mathcal{G}_{N-1}) = M_{N-1}$ almost surely, for $N \in \mathbb{N}$.
\end{proof}

\subsection{Proof of Proposition 3}
\begin{proof}
For part (i), assume that $X$ and $Y$ are not independent. Then there exists $(x_0, y_0)$ such that $P(X \leq x_0, Y \leq y_0) \neq F(x_0)G(y_0)$. Let $r_0 = F(x_0)$, $s_0 = G(y_0)$, and note that $0 < r_0, s_0 < 1$. Let $x_0' = \inf(x \in \mathbb{R}\colon F(x) > r_0)$, $y_0' = \inf(y \in \mathbb{R}\colon G(y) > s_0)$. By right-continuity we know that $F(x_0') = F(x_0)$ and $G(y_0') = G(y_0)$, and also $P(X \leq x_0', Y \leq y_0') = P(X \leq x_0, Y \leq y_0)$. By definition of $x_0'$, we know that $X \leq x_0'$ holds if and only if $F(X) \leq r_0$, which is equivalent to $\tilde{R} = F(X-) + U(F(X)-F(X-)) \leq r_0$, and analogously for $Y$ and $\tilde{S}$. So
\[
    P(\tilde{R} \leq r_0, \tilde{S} \leq s_0) = P(X \leq x_0', Y \leq y_0') \neq r_0s_0,
\]
which implies that the distribution $Q$ of $(\tilde{R},\tilde{S})$ is not uniform on $[0,1]^2$.

For part (ii), assume that $d_2 = d_1 K$ for an integer $K > 1$. Every bin $B_{k\ell}^{d_1}$ in the coarser grid contains exactly $K^2$ bins from the finder grid $B_{k'\ell'}^{d_2}$, $k',\ell'=0,\dots,d_2-1$, and we define $S(k,\ell) = \{(k',\ell')\colon B_{k'\ell'}^{d_2} \subseteq B_{k\ell}^{d_1}\}$. Then,
\begin{align*}
    \kl(Q^{(d_2)}\| \mathcal{U}) & = \sum_{k,\ell=0}^{d_2-1} d_2^2 q_{k\ell}^{(d_2)}\log(d_2^2 q_{k\ell}^{(d_2)})/d_2^2\\
    & = \sum_{k,\ell=0}^{d_1 - 1} K^2 \sum_{(k',\ell') \in S(k,l)} q_{k'\ell'}^{(d_2)} \log(q_{k'\ell'}^{(d_2)})/ K^2\\
    & \geq \sum_{k,\ell=0}^{d_1 - 1} K^2 \left(\sum_{(k',\ell') \in S(k,l)} q_{k'\ell'}^{(d_2)} / K^2\right) \log\left(\sum_{(k',\ell') \in S(k,l)}q_{k'\ell'}^{(d_2)} d_2^2 / K^2\right) \\
    & =  \sum_{k,\ell=0}^{d_1-1} q_{k\ell}^{(d_1)}\log(d_1^2 q_{k\ell}^{(d_1)}) \\
    & = \kl(Q^{(d_1)}\| \mathcal{U}).
\end{align*}
The inequality above holds due to Jensen's inequality applied to $x\mapsto x\log(x)$, and the second last equality uses the fact that $q_{k\ell}^{(d_1)} = Q(B_{k\ell}^{d_1}) = \sum_{(k',\ell') \in S(k,l)} q_{k'\ell'}^{(d_2)}$.

For part (iii), notice that
\[
    \mathbb{E}_{Q}[\log\{q^{(d)}(R,S)\}] = \mathbb{E}_{Q^{(d)}}[\log\{q^{(d)}(R,S)\}] = \int_{[0,1]^2} q^{(d)}(r,s)\log\{q^{(d)}(r,s)\} \, d(r,s),
\]
because $q^{(d)}$ is constant on the rectangles $\Bkld$. If $(r_0, s_0)$ is a Lebesgue point of $q$, then
\[
    q^{(d)}(r_0,s_0)\log\{q^{(d)}(r_0,s_0)\} \rightarrow q(r_0,s_0)\log\{q(r_0,s_0)\}
\]
by Theorem 7.10 of \citet{Rudin1987}, because $q$ is Lebesgue integrable. Since the Lebesgue points of $q$ have measure $1$ and because $x\log(x) \geq -\exp(-1)$ for all $x \geq 0$, we can apply Fatou's Lemma to obtain
\begin{align*}
    \int_{[0,1]^2} q(r,s)\log\{q(r,s)\} \, d(r,s) & = \int_{[0,1]^2} \liminf_{d \rightarrow \infty} q^{(d)}(r,s)\log\{q^{(d)}(r,s)\} \, d(r,s) \\
    & \leq \liminf_{d \rightarrow \infty} \int_{[0,1]^2} q^{(d)}(r,s)\log\{q^{(d)}(r,s)\} \, d(r,s) \\
    & \leq \limsup_{d \rightarrow \infty} \int_{[0,1]^2} q^{(d)}(r,s)\log\{q^{(d)}(r,s)\} \, d(r,s) \\
    & =  \limsup_{d \rightarrow \infty} \sum_{k,\ell=0}^{d} q_{k\ell}^{(d)}\log(d^2q_{k\ell}^{(d)}) \\
    & = \limsup_{d \rightarrow \infty} \sum_{k,\ell=0}^{d} \int_{\Bkld}q(r,s)\,d(r,s)\log\left\{d^2\int_{\Bkld}q(r,s)\,d(r,s)\right\} \\
    & \leq \limsup_{d \rightarrow \infty} \sum_{k,\ell=0}^{d} \int_{\Bkld}q(r,s)\log\left\{q(r,s)\right\}\,d(r,s) \\
    & = \int_{[0,1]^2} q(r,s)\log\{q(r,s)\} \, d(r,s),
\end{align*}
where the second last line applies Jensen's inequality, analogously to the proof of part (i). Notice that part (ii) would also follow from Theorem 21 of \citet{vanErven2014} if the sequence of grid sizes $(d_k)_{k \in \mathbb{N}}$ was such that $d_{k+1}$ is a multiple of $d_k$ for all $k \in \mathbb{N}$.

For the last part, if $Q$ is not absolutely continuous with respect to the Lebesgue measure, then there exist $\delta \in [0,1)$, a Lebesgue density $p$, and a probability measure measure $\nu$ singular with respect to the Lebesgue measure on $[0,1]^2$ such that for all Lebesgue measurable sets $B \subseteq [0,1]^2$,
\[
    Q(B) = \delta \int_B p(r,s) \, d(r,s) + (1-\delta)\nu(B).
\]
Let $A$ be the set on which $\nu$ is concentrated, so $\mathrm{Leb}(A) = 0$ and $\nu(A) = 1$. Let $\epsilon > 0$. Since $A$ is Lebesgue measurable with $\mathrm{Leb}(A) = 0$, it can be covered by a countable collection of rectangles whose union has at most Lebesgue measure $\epsilon$. We can take a finite subcollection $A_i$, $i = 1, \dots, M$, of these rectangles such that $Q(\bigcup_{i=1}^M A_i) \geq 1-\delta'$ for some $\delta' > \delta$.
These $M$ potentially overlapping rectangles can be written as a disjoint union of $M'$ rectangles $A'_i$, $i = 1, \dots, M'$, each with side lengths $\alpha_m$ and $\beta_m$, $m = 1, \dots, M'$. The Lebesgue measure of these rectangles satisfies
\begin{align*}
    \mathrm{Leb}\left(\bigcup_{i=1}^{M'}A'_i\right) = \sum_{m=1}^{M'} \alpha_m\beta_m \leq \epsilon.
\end{align*}
For $d$ sufficiently large, we have $\min(\alpha_m, \beta_m) > 1/d$, $m = 1, \dots, M'$.  A rectangle with side lengths $\alpha_m$ and $\beta_m$ can be covered by at most $M'' = (\lceil \alpha_m/d \rceil + 1)(\lceil \beta_m / d \rceil + 1)$ squares $A''_i$, $i = 1, \dots, M''$, with side length $1/d$. Since
\[
    \sum_{m=1}^{M'}\alpha_m = \sum_{m=1}^{M'}\alpha_m\beta_m / \beta_m \leq \epsilon d,
\]
which analogously holds for the sum over $\beta_1, \dots, \beta_{M'}$, the total Lebesgue measure of the covering squares is bounded by 
\begin{align*}
    \frac{1}{d^{2}}\sum_{m=1}^{M'} (\lceil \alpha_m/d \rceil + 1)(\lceil \beta_m / d \rceil + 1)
    & \leq \frac{1}{d^{2}} \sum_{m=1}^{M'} (\alpha_m / d + 2)(\beta_m / d + 2) \\
    & \leq \frac{1}{d^{2}} \Big( \epsilon/d^2 + 4 \epsilon d + 4M' \Big).
\end{align*}
Let $A'' = \bigcup_{i=1}^{M''}A''_i$. Applying Jensen's inequality, like in the proof of (ii), we obtain
\begin{align*}
    \mathbb{E}_Q[\log\{q^{(d)}(R,S)\}] & = \mathbb{E}_{Q^{(d)}}[\log\{q^{(d)}(R,S)\}] \\
    & \geq Q(A'') \log\left\{\frac{Q(A'')}{\mathrm{Leb}(A'')}\right\} + \{1-Q(A'')\} \log\left\{\frac{1-Q(A'')}{1-\mathrm{Leb}(A'')}\right\}.
\end{align*}
Since $\mathrm{Leb}(A'') \leq \epsilon$, $Q(A'') \geq 1-\delta'$, and $x\log(x) \geq -\exp(-1)$ for $x \geq 0$, we have 
\[
    \mathbb{E}_Q[\log\{q^{(d)}(R,S)\}] \geq (1-\delta')\log\left(\frac{d}{\epsilon/d^3 + 4\epsilon + 4M'/d}\right) -2\exp(-1) \rightarrow \infty, \ d \rightarrow \infty.
\]
\end{proof}

\subsection{Proof of Theorem 2}
For the proof of Theorem 2, we introduce the following quantities,
\begin{align*}
    &\MhatN = \prod_{k,\ell=0}^{d-1} \Big(d^{2} \frac{\bklN}{N} \Big)^{\bklN} = \frac{d^{2N}}{N^{N}} \prod_{k,\ell=0}^{d-1} \bklN^{\bklN}, \\
    &\btildeklN = \sum_{n=1}^{N} \Big\{(\Rtilden,\Stilden) \in \Bkld\Big\}, \\ 
    &\MtildeN = \prod_{k,\ell=0}^{d-1} \Big(d^{2} \frac{\btildeklN}{N} \Big)^{\btildeklN}
    = \frac{d^{2N}}{N^{N}} \prod_{k,\ell=0}^{d-1} \btildeklN^{\btildeklN}.
\end{align*}
Here $\MhatN$ is as \eqref{eq:test_martingale_d} but with the actual frequencies observed up to time $N$ instead of $f_n^{(d)}$; $\btildeklN$ are the bin counts with the probability integral transform \eqref{eq:honest_ranks}; and $\MhatN$ is as $\MhatN$ but with the probability integral transform replacing the sequential ranks. We do not indicate $d$ in $\MhatN$, $\MtildeN$ and omit it in quantities like $\qkld$, $\Qd$, to simplify notation.

\medskip
\begin{lemma}\label{lemma:sequential:full:lrt}
    For all $N \in \mathbb{N}$,
    \begin{align*}
        \log\Big( \frac{\MhatN}{\MN} \Big) \leq \frac{3(d^2-1)}{2} \log(N) + 5d^2.
    \end{align*}
\end{lemma}
\begin{proof}[of Lemma \ref{lemma:sequential:full:lrt}]
We have
\[
    \frac{\MhatN}{\MN} = \frac{(N-1+d^2)!}{(d^2-1)!N^N}\prod_{k,\ell=0}^{d-1}\frac{\bklN^{\bklN}}{(\bklN)!},
\]
and applying Stirling's formula, in the version of Lemma 10 by \citet{Duembgen2021}, to $\bklN^{\bklN}/(\bklN)!$ for which $\bklN > 0$ gives
\begin{align*}
    \log\left\{\frac{\bklN^{\bklN}}{(\bklN)!}\right\} & \leq \bklN\log(\bklN) - \log(2\pi)/2 - (\bklN - 1/2)\log(\bklN) \\
    & \qquad + \bklN - (12\bklN + 1)^{-1} \\
    & = -\log(2\pi)/2 + \log(\bklN)/2 + \bklN - (12\bklN + 1)^{-1}.
\end{align*}
Hence, with $K = \#\{(k,\ell)\colon \bklN > 0\} \leq d^2$, we have
\begin{align}
    \log\left\{\prod_{k,\ell=0}^{d-1}\frac{\bklN^{\bklN}}{(\bklN)!}\right\} & \leq -K\log(2\pi)/2 + \sum_{k,\ell\colon \bklN \neq 0} \log(\bklN)/2 + N - 1 - \frac{K}{12N+1} \nonumber \\
    & \leq -K\log(2\pi)/2 + d^2\log(N)/2 + N - 1 - \frac{K}{12N+1} \label{eq:stirling_bkln},
\end{align}
using that $\sum_{(k,\ell)\colon \bklN \neq 0} \bklN = N - 1$. Furthermore, 
\begin{align}
    \log\left\{\frac{(N-1+d^2)!}{(d^2-1)!N^N}\right\} & \leq -N\log(N) + (N-3/2+d^2)\log(N-1+d^2) \nonumber \\
    & \qquad -(N-1+d^2) + \frac{1}{12(N-1+d^2)} \label{eq:stirling_rest} \\
    & \qquad - (d^2-3/2)\log(d^2-1) + (d^2-1) - \frac{1}{12(d^2-1)+1} \nonumber
\end{align}
In a next step, we collect and bound all terms in \eqref{eq:stirling_bkln} and \eqref{eq:stirling_rest} that depend on $N$,
\begin{align}
    & d^2\log(N)/2 + N - \frac{K}{12N+1} -N\log(N) + (N-3/2+d^2)\log(N-1+d^2) \nonumber\\
    & \quad -(N-1+d^2) + \frac{1}{12(N-1+d^2)} \nonumber \\
    & = \log(N)\left(d^2/2 - N + N - 3/2 + d^2\right) + (N-3/2+d^2)\log\{1+(d^2-1)/N\} \nonumber \\
    & \quad - \frac{K}{12N + 1} + \frac{1}{12(N-1+d^2)} \nonumber \\
    & \leq 3(d^2-1)\log(N)/2 + (N-3/2)(d^2-1)/N + d^2\log\{1+(d^2-1)/N)\} + 1\nonumber \\
    & \leq 3(d^2-1)\log(N)/2 + d^2 + d^2\log(1+d^2) + 1 \label{eq:n_terms}.
\end{align}
Now we collect and bound all the terms in \eqref{eq:stirling_bkln} and \eqref{eq:stirling_rest} that do not depend on $N$,
\begin{align} 
    & -K\log(2\pi)/2 -1 +d^2-1 -(d^2-3/2)\log(d^2-1) - \frac{1}{12(d^2-1)+1} \nonumber \\
    & \leq d^2 - d^2\log(d^2-1) -2 - \log(2\pi)/2 + 3\log(d^2-1)/2 \label{eq:not_n_terms}
\end{align}
Combining \eqref{eq:n_terms} and \eqref{eq:not_n_terms}, we get the following bound for $\log(\MhatN/\MN)$,
\begin{align*}
    & 3(d^2-1)\log(N)/2 + 2d^2 + d^2\log(1+d^2) - d^2\log(d^2-1) + 3\log(d^2-1)/2 \\
    & \quad - 1 - \log(2\pi)/2 \\
    & \leq 3(d^2-1)\log(N)/2 + 2d^2 + d^2\log\{1+2/(d^2-1)\} + 3\log(d^2-1)/2 \\
    & \leq 3(d^2-1)\log(N)/2 + 5d^2
\end{align*} 
using that $d \geq 2$ and $\log\{1+2/(2^2-1)\} \leq 1$.
\end{proof}

For the following Lemmata, we define $\hat{q}_{k\ell} = \bkln/N$, $k,\ell=0,\dots,d-1$. The subscript $N$ is omitted whenever it is not necessary for the understanding.

\medskip
\begin{lemma}\label{lemma:sequential:mle}
    If $N \geq 3$, then 
    \begin{align*}
        \pr \Big\{ \max_{k,\ell = 0, \dots, d-1} |\qhatkl - \qkl| \leq \frac{40 d \log(N)^{1/2}}{N^{1/2}} \Big\} \geq 1 - 4\exp\Big\{ -127 d^{2} \log(N) \Big\}.
    \end{align*}
\end{lemma}
\begin{proof}[of Lemma \ref{lemma:sequential:mle}]
    Recall the definition of the probability integral transforms,
    \begin{align*}
        &\Rtilden = \Un F(\Xn) + (1-\Un)F(\Xn-), \\
        &\Stilden = \Vn G(\Yn) + (1-\Vn)G(\Yn-),
    \end{align*}
     and let $\qtildekl = \btildeklN / N$ denote the maximum likelihood estimator of $\qkl$ using the $\Rtilden$, $\Stilden$, $n = 1, \dots, N$ as observations.  Therefore, since
    \begin{align*}
        \max_{k,\ell = 0, \dots, d-1} |\qhatkl - \qkl| \leq \max_{k,\ell = 0, \dots, d-1} |\qhatkl - \qtildekl| + \max_{k,\ell = 0, \dots, d-1} |\qtildekl - \qkl|,
    \end{align*}
    it suffices to bound each of the two terms on the right hand side separately.
    
    To this end, define the random variables
    \begin{align*}
        \Kn = \ONE\Big(\lfloor d\Rtilden \rfloor \neq \lfloor d\Rn \rfloor\Big), \ \Ln = \ONE\Big(\lfloor d\Stilden \rfloor \neq \lfloor d\Sn \rfloor\Big), \ \Zn = \max(\Kn, \Ln).
    \end{align*}
    Thus, $\Zn$ is an indicator for whether the sequential rank and the probability integral transform fall into different bins.  Temporarily fix a value $n \in \Nbb$. For a sequence of positive numbers $\epsilonn$, define the events
    \begin{align*}
        \Ascr_{n} = \Big\{ \sup_{x \in \Rbb} |\Fhatn(x) - F(x)| \leq \epsilonn \Big\}.
    \end{align*}
    The Dvoretzky-Kiefer-Wolfowitz inequality implies that, for any $\epsilon > 0$, 
    \begin{align*}
        \pr\Big\{ \sup_{x \in \Rbb} |\Fhatn(x) - F(x)| > \epsilon \Big\} \leq 2 \exp(-2n\epsilon^{2}).
    \end{align*}
    On $\Ascr_{n}$, it follows that 
    \begin{align*}
        |\Rn - \Rtilden| 
        &= | \Un\{\Fhatn(\Xn) - F(\Xn)\} + (1 - \Un)\{\Fhatn(\Xn-) - F(\Xn-)\} | \\
        &\leq \Un | \Fhatn(\Xn) - F(\Xn) | + (1 - \Un) |\Fhatn(\Xn-) - F(\Xn-)| \\ 
        &\leq \epsilonn.
    \end{align*}
    Since $\Rtilden$ are uniform random variables, we have
    \begin{align*}
        \pr(\Kn = 1, \Ascr_{n}) 
        & = \sum_{k = 0}^{d-1} \pr\big\{ \Rn \in (k/d,(k+1)/d], \Rtilden \not\in (k/d,(k+1)/d], \Ascr_{n} \big\} \\
        & \leq \sum_{k = 0}^{d-1} \pr\big\{ \Rn \in (k/d,(k+1)/d], \\
        & \qquad\qquad\quad \Rtilden \in (k/d - \epsilonn,k/d] \cup ((k+1)/d, (k+1)/d + \epsilonn] \big\} \\
        &\leq 2 d \epsilonn,
    \end{align*}
    implying that $\pr(\Kn = 1) \leq 2d\epsilonn + 2 \exp(-2n\epsilonn^{2})$.
    Letting $\epsilonn = \{\log(n) / n\}^{1/2}$, then 
    \begin{align*}
        \pr(\Kn = 1) \leq 2 d \left\{\frac{\log(n)}{n}\right\}^{1/2} + \frac{2}{n^{2}} \leq 4d \left\{\frac{\log(n)}{n}\right\}^{1/2}
    \end{align*}
    and 
    \begin{align*}
        \sum_{n=1}^{N} \pr(\Kn = 1) & \leq \sum_{n=1}^N 4d \left\{\frac{\log(n)}{n}\right\}^{1/2} \\
        & \leq  4d\log(N)^{1/2} \sum_{n=1}^N n^{-1/2} \\
        & \leq 8d\{N\log(N)\}^{1/2},
    \end{align*}
    and we define $\mu_N = 8d\{N\log(N)\}^{1/2}$.    
    The same arguments can be applied to $\Ln$. By Proposition \ref{prop:sequential_ranks} (ii), it follows that $\Kn$ and $\Ln$ are independent Bernoulli random variables. Therefore, it follows from Hoeffding's inequality \citep[Theorem 2.2.6]{vershynin2018high} that
    \begin{align*}
        \pr \Big( \sum_{n=1}^{N} \Zn \geq 4 \muN \Big) & 
        \leq \pr\Big(\sum_{n=1}^{N} \Kn \geq 2 \muN\Big) + \pr\Big(\sum_{n=1}^{N} \Ln \geq 2 \muN\Big) \\
        & \leq 2\exp \Big( - \frac{2\muN^{2}}{N} \Big) \\
        & \leq \exp\Big\{ -128 d^{2} \log(N) \Big\}.
    \end{align*}
    Thus, the above implies that 
    \begin{align*}
        \pr \Big( \max_{k,\ell = 0, \dots, d-1} | \bklN - \btildeklN | \leq 4 \muN \Big) \geq 1 - 2\exp\{-128 d^{2} \log(N)\},
    \end{align*}
    or, equivalently by dividing by $N$, 
    \begin{align*}
        \pr \Big( \max_{k,\ell = 0, \dots, d-1} | \qhatkl - \qtildekl | \leq \frac{4 \muN}{N} \Big) \geq 1 - 2\exp\{-128 d^{2} \log(N)\}.
    \end{align*}
    Next, note that $\qtildekl$ is the sample mean of $N$ independent Bernoulli random variables with probability of success $\qkl$.  By Hoeffding's inequality and a union bound, we have
    \begin{align*}
        \pr \Big( \max_{k,\ell=0, \dots, d-1} | \qtildekl - \qkl | \geq \frac{\muN}{N} \Big) & \leq 2 \exp \Big\{ -\frac{2\muN^{2}}{N} + \log(d^{2})  \Big\} \\
        & \leq 2 \exp\{ -128 d^{2} \log(N) +\log(d^2)\} \\
        & \leq 2 \exp\{ -127 d^{2} \log(N) \},
    \end{align*}
    using $N \geq 3$ in the last inequality. Thus, combining our above calculations yields
    \begin{align*}
        \pr \Big( \max_{k,\ell = 0, \dots, d-1} |\qhatkl - \qkl| \leq \frac{40d \log(N)^{1/2}}{N^{1/2}} \Big) \geq 1 - 4\exp\Big\{ -127 d^{2} \log(N) \Big\},
    \end{align*}
    which finishes the proof.
\end{proof}

For the next lemma, we recall that we use the convention $0\log(0) = 0$.

\medskip
\begin{lemma}[Log-Lipschitz]\label{lemma:loglipschitz}
    For $f(x) = x\log(x)$,
    \begin{align*}
        |f(x) - f(y)| \leq -(y-x)\log(y-x)
    \end{align*}
    for all $0 \leq x \leq y \leq e^{-1}$.  
\end{lemma}
\begin{proof}[of Lemma \ref{lemma:loglipschitz}]
Consider the function 
\begin{align*}
    g(y) = f(y) - f(x) - f(y-x).
\end{align*}
Note that $g(x) = 0$ and 
\begin{align*}
    g'(y) = f'(y) - f'(y-x) = \log(y) + 1 - \log(y-x) - 1 = \log(y) - \log(y-x) \geq 0
\end{align*}
for all $y \geq x$.  The last inequality follows from the monotonicity of $\log(\cdot)$.  Therefore, it follows that 
\begin{align*}
    f(y) - f(x) - f(y-x) \geq 0
\end{align*}
for all $0 \leq x \leq y$, implying that 
\begin{align*}
    |f(x) - f(y)| = f(x) - f(y) \leq -f(y-x) = -(y-x)\log(y-x)
\end{align*}
for $0 \leq x \leq y \leq e^{-1}$.
\end{proof}

\medskip
\begin{lemma}\label{lemma:sequential:lrt}
    If $N \geq 3$, $d \geq 2$, and $40d\log(N)^{1/2} \leq N^{1/2} e^{-1}$, then 
    \begin{align*}
        \pr \Big[ \log(\MhatN) \geq N \kl(\Q || \Ucal) - 40 d^{3} \{N\log(N)\}^{1/2} \Big] \geq 1 - 4\exp\Big\{ -127 d^{2} \log(N) \Big\}.
    \end{align*}
    If, in addition, $\kl(\Q || \Ucal) < 1/d^{4}$, then  
    \begin{align*}
        & \pr \Big( \log(\MhatN) \geq \Big[ \{N \kl(\Q || \Ucal)\}^{1/2} - \frac{40 d^{5} \log(N)^{1/2}}{1 - 1/\sqrt{2}} \Big] \{N \kl(\Q || \Ucal)\}^{1/2} \Big) \\
        & \quad \geq 1 - 4\exp\Big\{ -127 d^{2} \log(N) \Big\}.
    \end{align*}
\end{lemma}
\begin{proof}[of Lemma \ref{lemma:sequential:lrt}]
    Indeed, by definition, we have 
    \begin{align*}
        \log(\MhatN) = N \sum_{k,\ell=0}^{d-1} \qhatkl \log(\qhatkl d^{2}).
    \end{align*}
    From Lemma \ref{lemma:sequential:mle}, we have 
    \begin{align*}
        \pr \Big\{ \max_{k,\ell = 0, \dots, d-1} |\qhatkl - \qkl| \leq \frac{40 d \log(N)^{1/2}}{N^{1/2}} \Big\} \geq 1 - 4\exp\Big\{ -127 d^{2} \log(N) \Big\}.
    \end{align*}
    For the remainder of the proof, we restrict our attention to the above event.
    
    Now, letting $f(x) = x\log(x)$, 
    \begin{align*}
        \log(\MhatN) 
        &= N \sum_{k,\ell=0}^{d-1} f(\qhatkl) + N\log(d^{2}) \\ 
        &= N \sum_{k,\ell=0}^{d-1} f(\qkl) + N\log(d^{2}) + N \sum_{k,\ell=0}^{d-1} \Big\{ f(\qtildekl) - f(\qkl) \Big\} \\ 
        &\geq N \kl(\Q || \Ucal) - N \sum_{k,\ell=0}^{d-1} |f(\qtildekl - \qkl)| \\ 
        &\geq N \kl(\Q || \Ucal) + 40 d^{3} \{N \log(N)\}^{1/2} \log \Big\{ \frac{40d \log(N)^{1/2}}{N^{1/2}} \Big\} \\ 
        &\geq N \kl(\Q || \Ucal) - 40 d^{3} \{N\log(N)\}^{1/2},
    \end{align*}
    where we have applied Lemma \ref{lemma:loglipschitz} in the first inequality and the assumption $40d\log(N)^{1/2} \leq N^{1/2} e^{-1}$ in the second.  This yields the first claim. 

    For the second case, Pinsker's inequality implies that 
    \begin{align*}
        \max_{k,\ell=0,\dots,d-1} \Big| \qkl - \frac{1}{d^{2}} \Big| \leq \left\{\frac{1}{2} \kl(\Q || \Ucal )\right\}^{1/2},
    \end{align*}
    or, equivalently,
    \begin{align*}
        \max_{k,\ell=0,\dots,d-1} | \qkl d^{2} - 1 | \leq d^{2} \left\{\frac{1}{2} \kl(\Q || \Ucal )\right\}^{1/2}.
    \end{align*}
    Therefore, in the second case, the $\qkl$ are uniformly bounded away from zero.  Then, 
    \begin{align*}
        \log(\MhatN) 
        &\geq N \sum_{k,\ell=0}^{d-1} \qhatkl \log(\qkl d^{2}) \\ 
        &= N \sum_{k,\ell=0}^{d-1} \qkl \log(\qkl d^{2}) + N \sum_{k,\ell=0}^{d-1} \Big\{ (\qhatkl - \qkl) \log(\qkl d^{2}) \Big\} \\ 
        &\geq N \kl(\Q || \Ucal) - d^{2} N \Big(\max_{k,\ell=0,\dots, d-1} |\qhatkl - \qkl| \Big) \Big\{\max_{k,\ell=0,\dots, d-1} |\log(\qkl d^{2})| \Big\}.
    \end{align*}
    Using the inequality $x/(1+x) \leq \log(1+x) \leq x$ with $x = \qkl d^{2} - 1$, we have 
    \begin{align*}
        \max_{k,\ell=0,\dots, 1} |\log(\qkl d^{2})| 
        \leq \max\Big( \Big| \frac{\qkl d^{2} - 1}{\qkl d^{2}} \Big|, |\qkl d^{2} - 1| \Big)
        \leq \frac{d^{2} \{\kl(\Q || \Ucal) / 2\}^{1/2}}{1 - d^{2} \{\kl(\Q || \Ucal) / 2\}^{1/2}}.
    \end{align*}
    Combining our above calculations and using that $\kl(\Q || \Ucal) < d^{-4}$ yields
    \begin{align*}
        \log(\MhatN) 
        &\geq N \kl(\Q || \Ucal) - 40 d^{5} \{N \log(N)\}^{1/2} \frac{\{\kl(\Q || \Ucal) / 2\}^{1/2}}{1 - d^{2} \{\kl(\Q || \Ucal)/2\}^{1/2}} \\ 
        &\geq \Big[ \{N \kl(\Q || \Ucal)\}^{1/2} - \frac{40 d^{5} \log(N)^{1/2}}{1 - 2^{-1/2}} \Big] \{N \kl(\Q || \Ucal)\}^{1/2}.
    \end{align*}
    This finishes the proof.
\end{proof}

\medskip
\begin{proof}[of Theorem 2]
    Note that 
    \begin{align*}
        \pr(\tau \leq N) 
        &\geq \pr \Big( \MN \geq \frac{1}{\alpha} \Big) \\ 
        &\geq \pr \Big\{ \log\Big( \frac{\MN}{\MhatN} \Big) + \log(\MhatN) \geq -\log(\alpha) \Big\} \\ 
        &\geq \pr \Big\{ \log(\MhatN) \geq -\log(\alpha) - \log\Big( \frac{\MN}{\MhatN} \Big) \Big\}.
    \end{align*}
    The result now follows by applying Lemmata \ref{lemma:sequential:full:lrt} and \ref{lemma:sequential:lrt}, using the fact that
    \[
        3(d^2-1)\log(N)/2 + 5d^2  \leq 7d^3\{N\log(N)\}^{1/2}
    \]
    for the first part.
\end{proof}

\subsection{Proof of Proposition 4}
\begin{proof}
    For $\mathcal{M}_1$, the result holds because $M_{N, 1} \geq w_{d^*}M_N^{d^*}$. For $\mathcal{M}_0$, note that 
    \begin{align*}
        \frac{d^{2}}{d^{2} + n - 1} \leq f_{n}^{(d)} \leq d^{2}
    \end{align*}
    almost surely. Define $f_n^{(0)} \equiv 1$.  Then, by Jensen's inequality, we have
    \begin{align*}
        \log(M_{N,0}) 
        = \sum_{n=1}^N \log\Big(\sum_{d = 0}^{\infty}w_{d}f_n^{(d)}\Big)
        & \geq \sum_{n=1}^{N} \sum_{d = 0}^{\infty} w_{d} \log\Big( f_{n}^{(d)} \Big) \\
        & = \sum_{d = 0}^{\infty} w_{d} \sum_{n=1}^{N} \log\Big( f_{n}^{(d)} \Big) \\
        & = \sum_{d = 0}^{\infty} w_{d} \log\Big( M_{N}^{(d)} \Big).
    \end{align*}
    Now, temporarily fix $d \in \Nbb$.  Then, 
    \begin{align*}
        \log\Big( M_{N}^{(d)} \Big) 
        = \log\Big( \frac{M_{N}^{(d)}}{\Mhat_{N}^{(d)}} \Big) + \log\Big( \Mhat_{N}^{(d)} \Big)
    \end{align*}
    Considering each of the two terms separately, Lemma \ref{lemma:sequential:full:lrt} implies
    \begin{align*}
        \log\Big( \frac{M_{N}^{(d)}}{\Mhat_{N}^{(d)}} \Big) \geq -3(d^2-1) \log(N)/2 -5 d^{2}.
    \end{align*}
    Moreover,
    \begin{align*}
        \log\Big( \Mhat_{N}^{(d)} \Big) 
        = N \sum_{k,\ell=0}^{d-1} \frac{b_{k\ell,N}^{(d)}}{N} \log\Big(\frac{b_{k\ell, N}^{(d)}}{N} d^{2} \Big)
        \geq 0
    \end{align*}
    since the summation is the Kullback-Leibler divergence between the empirical distribution and the uniform distribution.  Hence, there exists a constant $c > 0$ such that 
    \begin{align*}
        \sum_{d\neq \dstar} w_{d} \log\Big( M_{N}^{(d)} \Big) 
        \geq - \sum_{d\neq \dstar} w_{d} \Big\{ 3(d^2-1)\log(N)/2 + 5 d^{2}\Big\} 
        \geq -c \log(N).
    \end{align*}
    Combining the above calculations yields
    \begin{align*}
        \log(M_{N,0}) \geq w_{\dstar} \log\Big( M_{N}^{(\dstar)} \Big) + \sum_{d\neq \dstar}^{\infty} w_{d} \log\Big( M_{N}^{(d)} \Big) 
        \geq w_{\dstar} \log\Big( M_{N}^{(\dstar)} \Big) - c \log(N).
    \end{align*}
    Now, letting 
    \begin{align*}
        \mathcal{N} = \Big\{ N \in \Nbb : N \kl(\Q^{(\dstar)} || \Ucal) - 44 (\dstar)^{3} \{N\log(N)\}^{1/2}  \geq \frac{-\log(\alpha) + c \log(N)}{w_{\dstar}} \Big\},
    \end{align*}
    it follows that $\Ebb(\tau)$ is equal to
    \begin{align*}
        \sum_{N \in \mathcal{N}} \pr(\tau \geq N) + \sum_{N \in \mathcal{N}^\comp} \pr(\tau \geq N) 
        \leq |\mathcal{N}| + \sum_{N \in \mathcal{N}^\comp} 4 \exp\Big\{ -127 (\dstar)^{2} \log(N) \Big\} < \infty,
    \end{align*}
    which finishes the proof.
\end{proof}

\section{Sequential adaptation of the BET} \label{sec:sequential_bet}
As already mentioned, our binning approach for constructing the test martingale is inspired by the binary expansion test by \citet{Zhang2019}. In this section, we explain how the BET can be unified with our methods for a slightly different version of our sequential test. Since the BET requires a grid with size of a power of $2$, we assume $d = 2^k$ throughout this section.

The BET leverages the observation that the problem of independence testing can be reduced to simpler tests for so-called cross-interactions. We refer to \citet{Zhang2019} for the detailed derivation and the interpretation of these cross-interactions as interactions between the binary number expansions of $F(X)$ and $G(Y)$. For our purpose here, the following simplified description is sufficient: if one divides the $d^2$ cells $\Bkld$, $k,l=0,\dots,d-1$, into two groups
\begin{equation} \label{eq:cross_interaction}
    A_1^d = \bigcup_{k,\ell \in \mathcal{S}} \Bkld, \quad A_2^d = [0,1]^2\setminus A_1,
\end{equation}
with an index set $S \subset \{0,\dots,d-1\}^2$ such that $\#S = d^2/2$, then $(R_n, S_n)$ lies in $A_1$ and $A_2$ with probability $1/2$ each under $H_0$. Clearly, not every division $A_1, A_2$ is suitable for testing $H_0$; for example, by taking~$\mathcal{S} = \{(k,\ell)\colon k = 0, \dots, d^2/2-1\}$ one could not detect any dependence between $R_n$ and $S_n$. \citet[Theorem 4.1]{Zhang2019} shows that for testing independence on a $(d\times d)$-discretization of $[0,1]^2$, it is sufficient to consider only $(d-1)^2$ cross-interactions, each corresponding to a dichotomization of the form \eqref{eq:cross_interaction}; Figure \ref{fig:illustration_seq_bet} illustrates these for $d = 4$.

\begin{figure}[t]
    \centering
    \includegraphics[width=\textwidth]{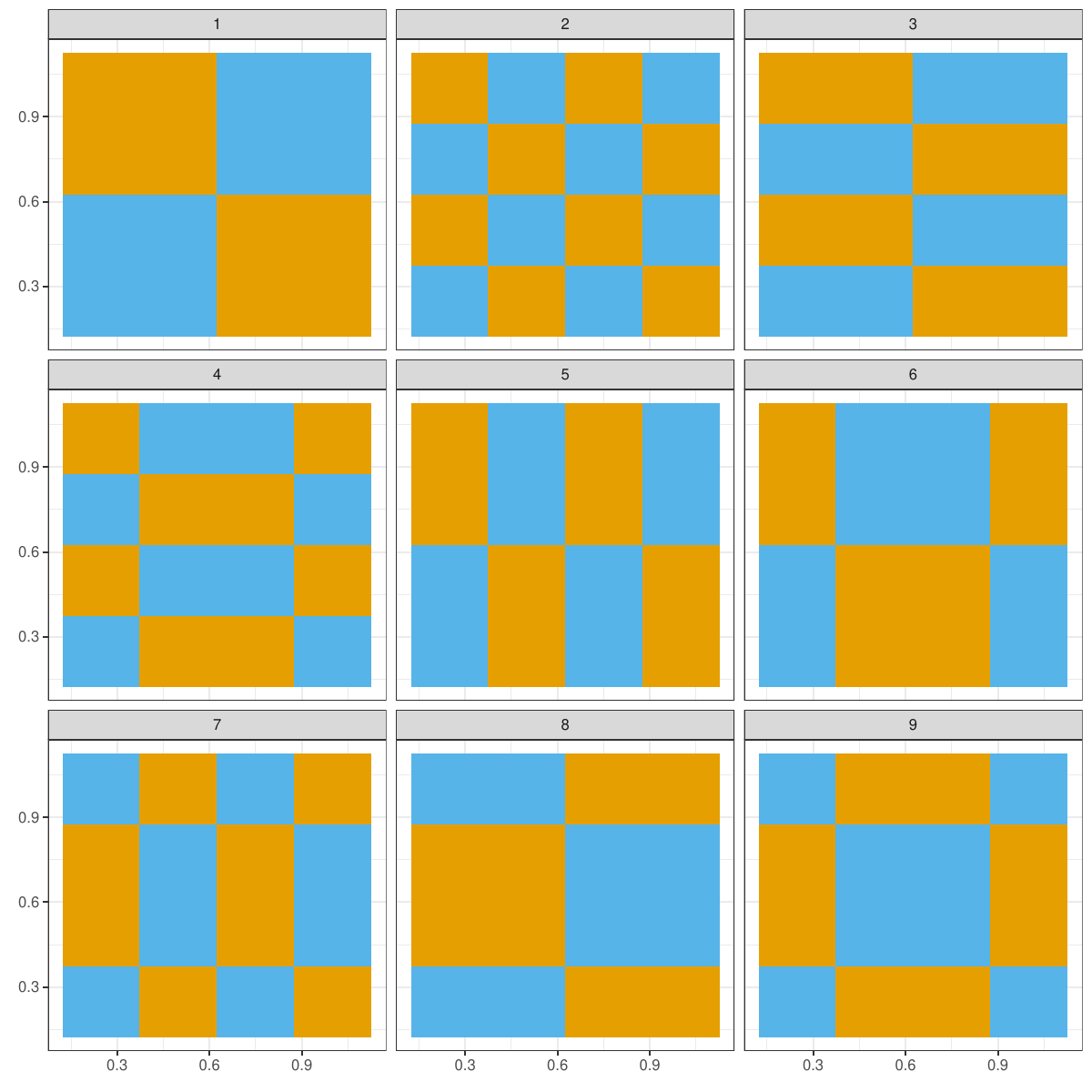}
    \caption{Interaction regions for the BET with $d = 4$.  \label{fig:illustration_seq_bet}}
\end{figure}

A test martingale for the hypothesis $P\{(R_n, S_n) \in A_1^{d}\} = 1/2$ can be constructed with exactly the same approach as in Section \ref{sec:construction}, replacing the $\Bkld$, $k,\ell=0,\dots,d-1$, by the two sets $A_1^{d}, A_2^{d}$. This yields $(d-1)^2$ test martingales $\mathcal{M}^{d,m} = (M_n^{(d,m)})_{n\in\mathbb{N}}$, $m = 1, \dots (d-1)^2$, one for each interaction encoded by a certain index set $\mathcal{S}_m$ in Equation \eqref{eq:cross_interaction}. The exact description of these index sets is in Section 3.3 of \citet{Zhang2019}, but it is not crucial for the understanding of our sequential version of the BET. The test martingales for all interactions can then be averaged pointwise to $M_N = \sum_{m=1}^{(d-1)^2}M_N^{(d,m)}/(d-1)^2$.

This sequential adaptation of the BET is guaranteed to have power if $\kl(\Qd\|\Ucal) > 0$. More precisely, in the proof of Theorem 2, only the number of bins on $[0,1]^2$ and the bin probabilities are relevant, but not the shape of the bins. Hence the statement of Theorem 2 also applies to the martingales $\mathcal{M}^{(d,m)}$, with the following adjustments. The Kullback-Leibler divergence $\kl(\Qd\|\Ucal)$ is replaced by the Kullback-Leibler divergence $\kl(\Q^{(d,m)}\|\Ucal)$ between the distribution with density
\begin{align*}
    q^{(d,m)}(r,s) = 2Q(A_1^{d,m})^{\mathbbm{1}\{(r,s)\in A_2^d\}}Q(A_2^{d,m})^{\mathbbm{1}\{(r,s)\in A_1^{d,m}\}}, \ (r,s)\in [0,1]^2;
\end{align*}
and the uniform distribution $\mathcal{U}$; and the discretization depth $d$ is replaced by $2$, since $\mathcal{M}^{d,m}$ is based on only two bins $A_1^{d,m}, A_2^{d,m}$.

Compared to our original approach, the advantage of the sequential BET is that each of the martingales $\mathcal{M}^{d,m}$ tests a simpler binomial hypothesis, which can be rejected with less data if the interaction captures the dependence between $X$ and $Y$. However, unless one has prior knowledge which interactions might be relevant, each martingale $\mathcal{M}^{d,m}$ only has weight $(d-1)^{-2}$ in the mixture $\sum_{m=1}^{(d-1)^2}M_N^{(d,m)}/(d-1)^2$. Moreover, it follows from the Jensen's inequality that $\kl(\Q^{(d,m)}\|\Ucal) \leq \kl(\Qd\|\Ucal)$; i.e.~the growth rate is smaller than that of our original test. Hence it depends on the specific application whether this or our original approach has more power.

\section{More implementation details} 
\subsection{Sequential rank computation} \label{sec:efficient_implementation}

The efficient computation of the normalized sequential ranks $\hat{F}_n(X_n)$ requires a data structure that allows to insert and order observations efficiently. Our implementation is based on the tree structure in the GNU Policy-Based Data Structures (\url{https://gcc.gnu.org/onlinedocs/libstdc++/ext/pb_ds/}), which allows inserting $X_{n+1}$ into the container for $X_1, \dots, X_n$ and querying its position with a complexity of $\log(n)$. Hence the overall complexity for computing the sequential ranks is of order $\mathcal{O}\{n\log(n)\}$. For a discretization depth $d$, the complexity for the computation of $M_n^{(d)}$ is $\mathcal{O}\{d^2n\log(n)\}$ if the correction for uniform marginals is applied, and $\mathcal{O}\{\log(d)n\log(n)\}$ for the estimator without this correction. In the latter case, the $\log(d)$ complexity is for finding the position of $(R_n, S_n)$ in the bins on $[0,1]^2$.

\subsection{Derandomization for distributions with atoms} \label{sec:derandomization}
The derandomization described in Section \ref{sec:implementation} relies on the fact that $R_n = \hat{F}_n(X_n)-U_n/n$, $S_n = \hat{G}_n(Y_n)-U_n/n$ in the case of continuous marginal; see also Remark \ref{remark:filtration}. Derandomization is more delicate if the marginal distributions of $(X,Y)$ are not continuous, because the variables $(U_n)_{n\in\mathbb{N}}$, $(V_n)_{n\in\mathbb{N}}$ are required to randomly break ties in the ranks, and a different strategy has to be employed.

A natural idea would be to simulate $B$ processes of randomization variables $\{U_n(b), V_n(b)\}_{n\in\mathbb{N}}$, $b = 1, \dots, B$, and then average the 
the corresponding test martingales $\{M_n(b)\}_{n\in\mathbb{N}}$ at each time point; here $\{M_n(b)\}_{n\in\mathbb{N}}$ can be a martingale for a single discretization $d$, or also one of the aggregated versions from Section \ref{sec:combination}. Unfortunately, this does not yield a martingale in general, because the processes $\{M_n(b)\}_{n\in\mathbb{N}}$ are measurable with respect to different filtrations. A valid method is to transform the test martingales to anytime-valid p-values, i.e.~$p_N(b) = \{\max{n = 1, \dots, N}M_n(b)\}^{-1}$, as in \citet{HenziZiegel2022} and \citet{Arnold2021}. Since $\{M_n(b)\}_{n\in\mathbb{N}}$, $b = 1, \dots, B$, are based on the same underlying data and only use different randomization, one can expect $p_N(b)$, $b = 1, \dots, B$, to be positively dependent. For such situations, merging by arithmetic or geometric mean,
\[
    p_N^A = \frac{2}{B}\sum_{b=1}^B \left\{\max_{n = 1, \dots, N}M_n(b)\right\}^{-1}, \quad p_N^G = \exp(1)\prod_{b=1}^B\left\{\max_{n = 1, \dots, N}M_n(b)\right\}^{-1/B},
\]
are powerful combination strategies \citep{Vovk2020}, which both guarantee $P(p_n^\xi \leq \alpha \text{ for some } n \in \mathbb{N}) \leq \alpha$.

\section{Additional figures, tables, and simulation results} \label{sec:moresim}
\subsection{Illustration of theoretical properties}
To illustrate the theoretical properties of our method, we display the test martingales for a single realization of the Linear simulation example from Table \ref{tab:simulations} with $l = 5$. Figure \ref{fig:illustration_small_large_n} shows test martingales for $d = 2, 4, 8, 16$, as well as their combinations with $\eta = 0,1$. In Figure \ref{fig:kl_growth_rate}, we depict the empirical growth rates $\sum_{i=1}^n \log\{f_n^{(d)}(R_i, S_i)\}/n$ and the KL-divergence $\kl(\Qd\|\Ucal)$ for the same example, with noise levels $l = 1, 5, 9$. As predicted by our theory, for large $n$, the test martingale with $d = 16$ grows at the fastest rate, with growth rate converging to $\kl(Q^{(16)}\|\Ucal)$ as $n$ increases. The martingales with smaller $d$ start to grow at smaller sample sizes, but have a smaller growth rate. The combination method with $\eta = 0$ is superior to all other martingales at small sample sizes, but is eventually exceeded by $(M_n^{(4)})_{n\in\mathbb{N}}$ and by the average of the martingales ($\eta = 1$), but only at a very large sample size. Density estimates for the distribution of the sequential ranks are shown in Figure \ref{fig:illustration_densities}. As can be seen in Figure \ref{fig:illustration_simple_sinkhorn}, in this particular simulation the correction with Sinkhorn's algorithm is uniformly better than the simple estimators.

\begin{figure}[b]
    \centering
    \includegraphics[width=\textwidth]{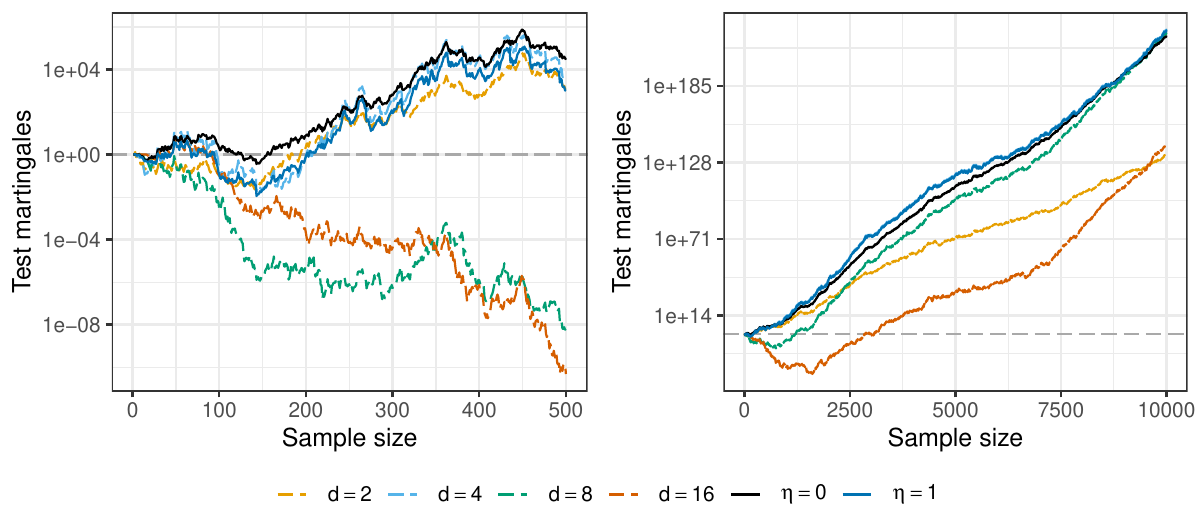}
    \caption{Test martingales and their combinations for the Linear simulation example with $l = 5$, for sample sizes up to $500$ (left panel) and $10\,000$ (right panel). \label{fig:illustration_small_large_n}}
\end{figure}

\begin{figure}[b]
    \centering
    \includegraphics[width=\textwidth]{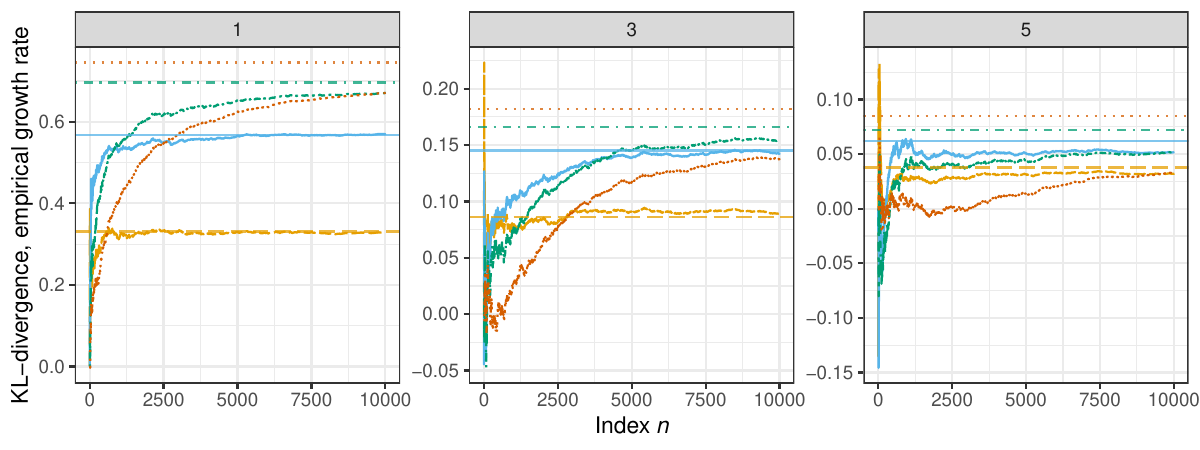}
    \caption{Empirical growth rate $\sum_{i=1}^n \log\{f_n^{(d)}(R_i, S_i)\}/n$ of realizations of the test martingales for the Linear example, together with the KL-divergence $\kl(\Qd\|\Ucal)$ as horizontal lines, for $l = 1, 5, 9$ (panels) and $d = 2, 4, 8, 16$ (dashed, solid, dot-dashed, dotted lines). \label{fig:kl_growth_rate}}
\end{figure}

\begin{figure}[b]
    \centering
    \includegraphics[width=\textwidth]{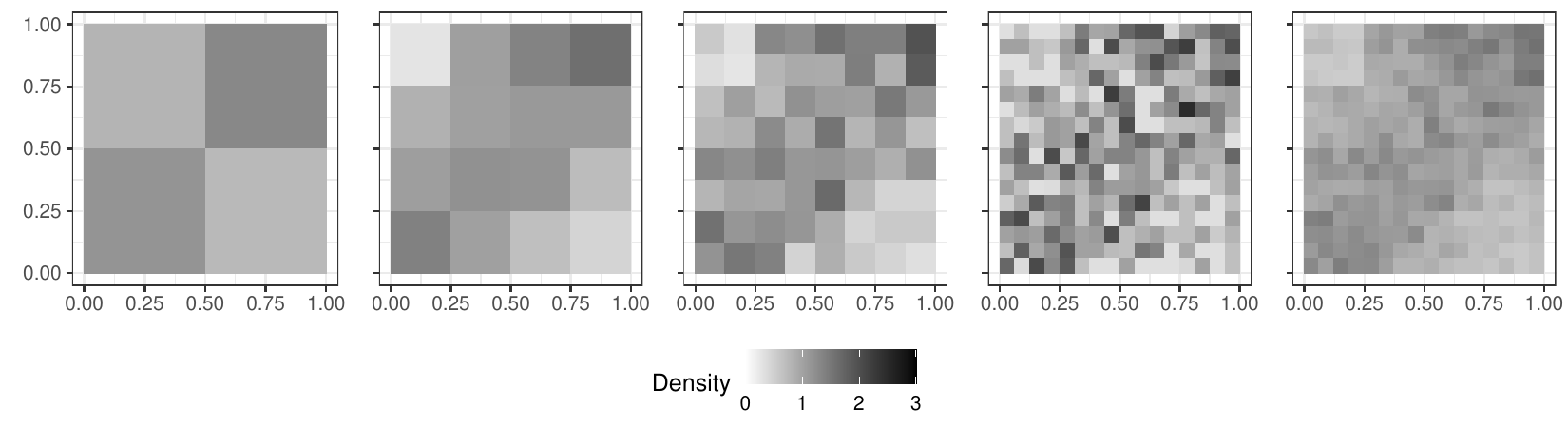}
    \caption{Density estimates for the Linear simulation example with $l = 5$, at $n = 500$. Then panels show $f_{500}^{(d)}$, $d = 2, 4, 8, 16$, as well as the aggregated density $\sum_{k=0}^4 f_{500}^{(2^k)}/5$, from left to right. \label{fig:illustration_densities}}
\end{figure}

\begin{figure}[b]
    \centering
    \includegraphics[width=\textwidth]{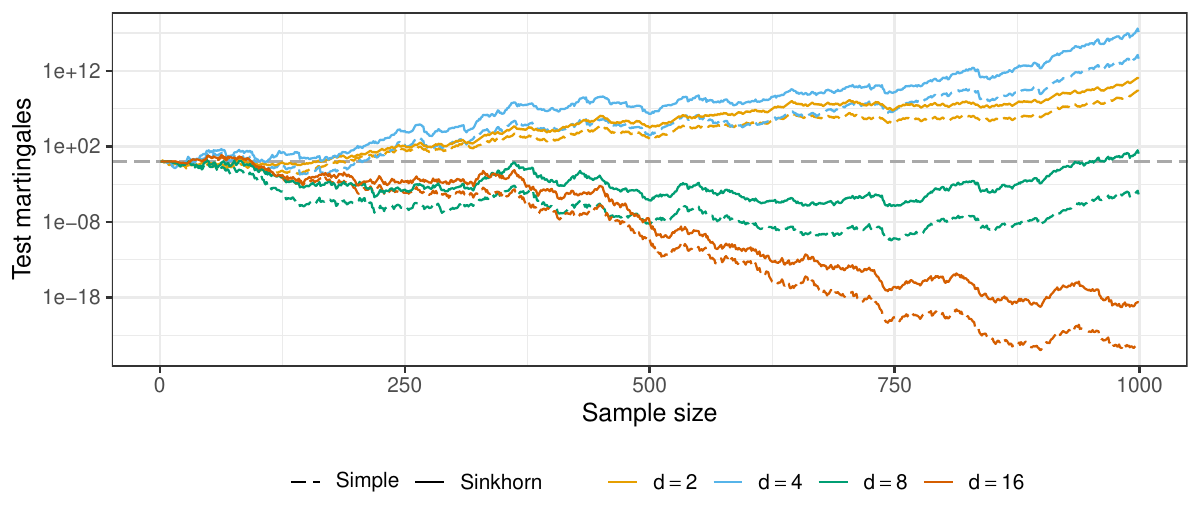}
    \caption{Test martingales as in Figure \ref{fig:illustration_small_large_n} for sample size up to $1000$, with ``Simple'' referring to the uncorrected martingales and ``Sinkhorn'' to the corrected version based on Sinkhorn's Algorithm. \label{fig:illustration_simple_sinkhorn}}
\end{figure}

\subsection{More simulation results}
Figure \ref{fig:rejection_rates_supplement} compares the rejection rates for all methods in the simulation examples from Section \ref{sec:simulations} and additional examples described in Table \ref{tab:simulations_suppl}, with noise levels $l = 1, 3, 5, 7, 9$. The conclusions are the same as in Section \ref{sec:simulations}. The sequential BET generally has less power than our other methods with $\eta = 0$, which we attribute to the fact that more observations are required to compensate the weight of $(d-1)^{-1}$ of each interaction term in the mixture. Figure \ref{fig:ville_illustration} shows the distribution function of $p_N = \{\max_{n=1,\dots,N}M_n\}^{-1}$ as defined in Section \ref{sec:ville_gap}. Table \ref{tab:seq_vs_nonseq_suppl} compares the truncated sequential test to the BET for the Checkerboard, Parabolic, and Sine simulation examples. Table \ref{tab:seq_vs_nonseq_suppl_256} is the same comparison for all simulation examples, but with the sequential test truncated at $n = 256$ instead of $n = 512$. In Tables \ref{tab:hoeffdingsd} and \ref{tab:chatterjee} we compare our truncated sequential test against the tests by \citet{Hoeffding1948} and \citet{Chatterjee2021}, for which we took the \textsf{R} implementations by \citet{EvenZohar2023} and \citet{Chatterjee2023}, respectively. The conclusions from these comparisons are the same as in Section \ref{sec:ville_gap}. The difference in power between our sequential test and the BET is smaller when the maximum sample size is $256$ instead of $512$, and our test generally has more power than the test by \citet{Hoeffding1948} in the Checkerboard, Circular, and Local examples and less in the others. The test by \citet{Chatterjee2021} has less power than the other methods, except for the Sine and Parabolic examples with $n = 64, 128$.

\begin{table}[b]
    \vspace{3cm}
	\centering
	\begin{tabular}{l l l}
		Scenario        &   Generation of $X$   &   Generation of $Y$   \\[0.25em]
		Checkerboard    & $X = W + \epsilon$    & $Y = \ONE(W = 2)(V_1+4\epsilon') + \ONE(W\neq 2)(V_2+4\epsilon'')$    \\[0.25em]
		Parabolic       & $X = U$               & $Y = (X - 0.5)^2 + 1.5\epsilon$    \\
		Sine            & $X = U$               & $Y = \sin(4\pi X) + 8\epsilon$    \\[0.5em]
	\end{tabular}
    \caption{Simulation examples, with the following variables: for $l = 1, 2, \dots, 10$, $\epsilon, \epsilon', \epsilon'' \sim \mathcal{N}\{0,(l/40)^2\}$; $U \sim \mathrm{Unif}[0,1]$; $W \sim \mathrm{Unif}\{1, 2, 3\}$; $V_1 \sim \mathrm{Unif}\{2, 4\}$; $V_2\sim\mathrm{Unif}\{1,3,5\}$.    \label{tab:simulations_suppl}}
    \vspace{1cm}
\end{table}

\begin{figure}[t]
    \centering
    \includegraphics[width=\textwidth]{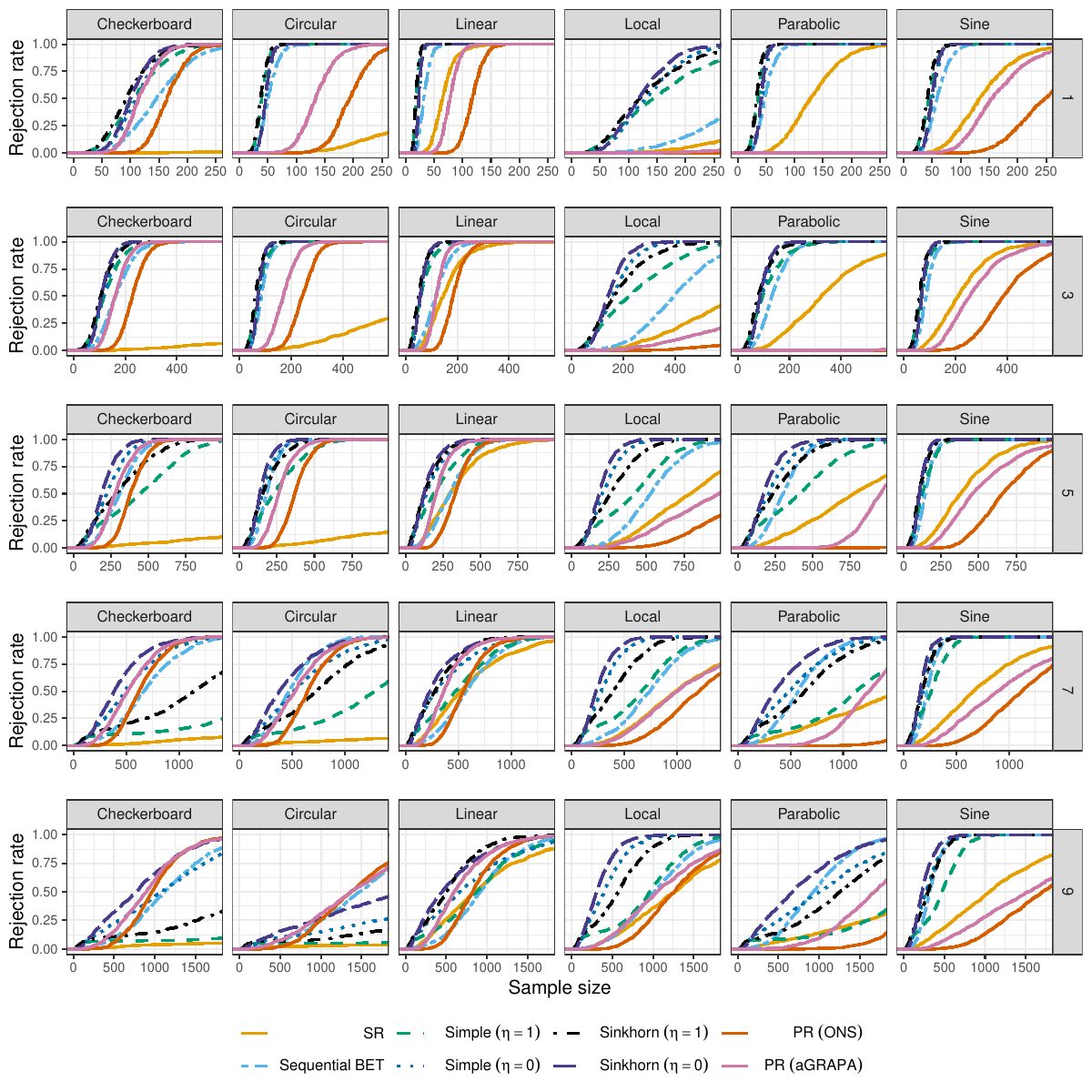}
    \caption{Rejection rates of the different methods for noise levels $l = 1, 3, 5,7, 9$ (panel rows). \label{fig:rejection_rates_supplement}}
\end{figure}

\begin{figure}
    \centering
    \includegraphics[width=\textwidth]{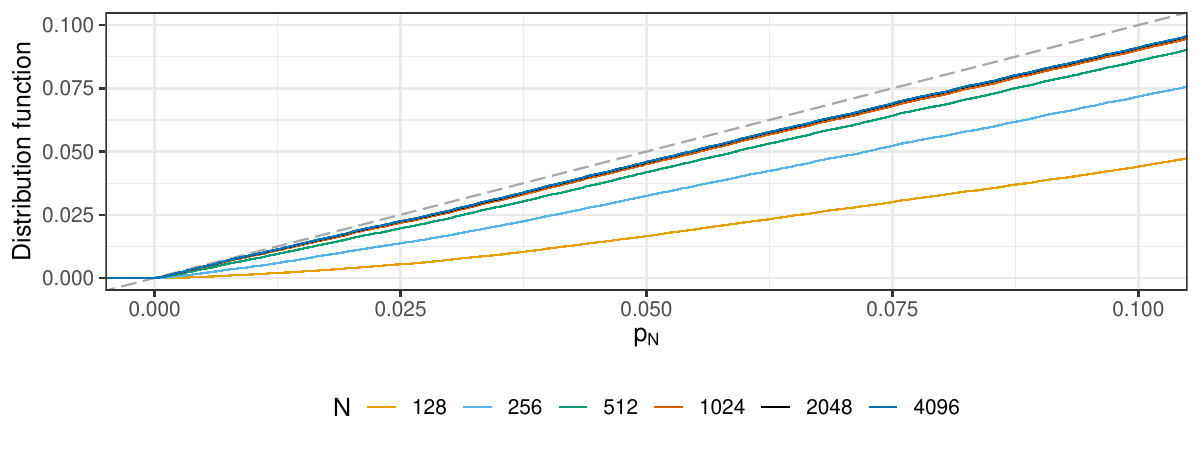}
    \caption{Distribution function of $p_N = (\max_{n=1,\dots,N}M_n)^{-1}$ as defined in Section \ref{sec:ville_gap}, estimated on $100\,000$ simulations. \label{fig:ville_illustration}}
\end{figure}

\begin{table}[t]
    \centering
    \small
	\begin{tabular}{llllllllcc}
		Simulation $\quad$ & $l \quad$ & \multicolumn{2}{l}{$\kl(Q^{(d)}\|\Ucal)\quad$} & \multicolumn{4}{l}{BET with $n$ samples $\qquad\qquad$} & \multicolumn{2}{c}{Sequential rank test} \\[0.5em]
	& & 2	& 8 & 64 & 128 & 256 & 512 & Power & Mean sample size\\[0.25em]

 & 1 & 0.00 & 0.31 & 0.45 & 0.93 & 1.00 & 1.00 & 1.00 & 100\\
 & 3 & 0.00 & 0.26 & 0.33 & 0.88 & 1.00 & 1.00 & 1.00 & 110\\
Checkerboard & 5 & 0.00 & 0.10 & 0.10 & 0.35 & 0.90 & 1.00 & 1.00 & 195\\
 & 7 & 0.00 & 0.04 & 0.04 & 0.10 & 0.37 & 0.78 & 0.71 & 353\\
 & 9 & 0.00 & 0.02 & 0.04 & 0.06 & 0.21 & 0.48 & 0.35 & 429\\[0.1em]
 & 1 & 0.00 & 0.60 & 1.00 & 1.00 & 1.00 & 1.00 & 1.00 & 40\\
 & 3 & 0.00 & 0.16 & 0.48 & 0.89 & 1.00 & 1.00 & 1.00 & 89\\
Parabolic & 5 & 0.00 & 0.06 & 0.12 & 0.35 & 0.88 & 1.00 & 0.99 & 186\\
 & 7 & 0.00 & 0.03 & 0.07 & 0.14 & 0.51 & 0.87 & 0.71 & 332\\
 & 9 & 0.00 & 0.02 & 0.03 & 0.09 & 0.28 & 0.57 & 0.40 & 415\\[0.1em]
 & 1 & 0.00 & 0.56 & 1.00 & 1.00 & 1.00 & 1.00 & 1.00 & 49\\
 & 3 & 0.00 & 0.30 & 0.91 & 1.00 & 1.00 & 1.00 & 1.00 & 69\\
Sine & 5 & 0.00 & 0.15 & 0.47 & 0.93 & 1.00 & 1.00 & 1.00 & 108\\
 & 7 & 0.00 & 0.09 & 0.19 & 0.61 & 0.99 & 1.00 & 1.00 & 160\\
 & 9 & 0.00 & 0.06 & 0.11 & 0.35 & 0.90 & 1.00 & 0.96 & 238\\[0.5em]
	\end{tabular}
    \caption{Like Table \ref{tab:seq_vs_nonseq}, but for the Checkerboard, Parabolic, and Sine example. \label{tab:seq_vs_nonseq_suppl}}
\end{table}

\begin{table}[t]
    \centering
    \small
    \begin{tabular}{lllllcc}
Simulation & $l$ & \multicolumn{3}{c}{BET with $n$ samples} & \multicolumn{2}{c}{Sequential rank test} \\[0.5em]

& & 64 & 128 & 256 & Power & Mean\\[0.25em]
 & 1 & 0.45 & 0.93 & 1.00 & 1.00 & 97\\
 & 3 & 0.33 & 0.88 & 1.00 & 1.00 & 106\\
Checkerboard & 5 & 0.10 & 0.35 & 0.90 & 0.81 & 173\\
 & 7 & 0.04 & 0.10 & 0.37 & 0.34 & 224\\
 & 9 & 0.04 & 0.06 & 0.21 & 0.19 & 237\\[0.25em]
 
 & 1 & 1.00 & 1.00 & 1.00 & 1.00 & 44\\
 & 3 & 0.88 & 1.00 & 1.00 & 1.00 & 69\\
Circular & 5 & 0.28 & 0.72 & 1.00 & 0.94 & 135\\
 & 7 & 0.06 & 0.15 & 0.59 & 0.37 & 218\\
 & 9 & 0.02 & 0.04 & 0.13 & 0.12 & 243\\[0.25em]
 
 & 1 & 1.00 & 1.00 & 1.00 & 1.00 & 22\\
 & 3 & 0.75 & 0.99 & 1.00 & 1.00 & 56\\
Linear & 5 & 0.25 & 0.63 & 0.96 & 0.92 & 123\\
 & 7 & 0.12 & 0.33 & 0.66 & 0.62 & 181\\
 & 9 & 0.07 & 0.18 & 0.44 & 0.35 & 215\\[0.25em]
 
 & 1 & 0.12 & 0.44 & 0.95 & 0.99 & 116\\
 & 3 & 0.08 & 0.28 & 0.88 & 0.96 & 138\\
Local & 5 & 0.06 & 0.17 & 0.68 & 0.83 & 170\\
 & 7 & 0.05 & 0.13 & 0.45 & 0.62 & 193\\
 & 9 & 0.06 & 0.10 & 0.31 & 0.44 & 211\\[0.25em]
 
 & 1 & 1.00 & 1.00 & 1.00 & 1.00 & 39\\
 & 3 & 0.48 & 0.89 & 1.00 & 1.00 & 85\\
Parabolic & 5 & 0.12 & 0.35 & 0.88 & 0.81 & 161\\
 & 7 & 0.07 & 0.14 & 0.51 & 0.41 & 214\\
 & 9 & 0.03 & 0.09 & 0.28 & 0.22 & 232\\[0.25em]
 
 & 1 & 1.00 & 1.00 & 1.00 & 1.00 & 47\\
 & 3 & 0.91 & 1.00 & 1.00 & 1.00 & 66\\
Sine & 5 & 0.47 & 0.93 & 1.00 & 1.00 & 103\\
 & 7 & 0.19 & 0.61 & 0.99 & 0.92 & 149\\
 & 9 & 0.11 & 0.35 & 0.90 & 0.66 & 190\\[0.5em]

    \end{tabular}
    \caption{Like Table \ref{tab:seq_vs_nonseq} but with our sequential test truncated at $n = 256$. \label{tab:seq_vs_nonseq_suppl_256}}
\end{table}

\begin{table}[t]
    \centering
    \small
    \begin{tabular}{llllllcc}
Simulation & $l$ & \multicolumn{4}{c}{Hoeffdings's $D$ with $n$ samples} & \multicolumn{2}{c}{Sequential rank test} \\[0.5em]
& & 64 & 128 & 256 & 512 & Power & Mean\\[0.25em]

 & 1 & 0.12 & 0.18 & 0.34 & 0.97 & 1.00 & 100\\
 & 3 & 0.12 & 0.15 & 0.31 & 0.91 & 1.00 & 110\\
Checkerboard & 5 & 0.09 & 0.10 & 0.15 & 0.36 & 1.00 & 195\\
 & 7 & 0.08 & 0.10 & 0.11 & 0.20 & 0.71 & 353\\
 & 9 & 0.08 & 0.08 & 0.10 & 0.14 & 0.35 & 429\\[0.25em]
 
 & 1 & 1.00 & 1.00 & 1.00 & 1.00 & 1.00 & 45\\
 & 3 & 0.29 & 0.94 & 1.00 & 1.00 & 1.00 & 72\\
Circular & 5 & 0.06 & 0.13 & 0.59 & 1.00 & 1.00 & 144\\
 & 7 & 0.05 & 0.05 & 0.08 & 0.26 & 0.68 & 342\\
 & 9 & 0.05 & 0.04 & 0.03 & 0.06 & 0.17 & 467\\[0.25em]
 
 & 1 & 1.00 & 1.00 & 1.00 & 1.00 & 1.00 & 23\\
 & 3 & 1.00 & 1.00 & 1.00 & 1.00 & 1.00 & 59\\
Linear & 5 & 0.80 & 0.98 & 1.00 & 1.00 & 1.00 & 135\\
 & 7 & 0.53 & 0.82 & 0.98 & 1.00 & 0.86 & 252\\
 & 9 & 0.35 & 0.60 & 0.88 & 1.00 & 0.58 & 357\\[0.25em]
 
 & 1 & 0.24 & 0.48 & 0.88 & 1.00 & 1.00 & 121\\
 & 3 & 0.20 & 0.38 & 0.74 & 0.99 & 1.00 & 144\\
Local & 5 & 0.15 & 0.22 & 0.52 & 0.92 & 1.00 & 189\\
 & 7 & 0.13 & 0.22 & 0.38 & 0.80 & 0.95 & 248\\
 & 9 & 0.12 & 0.15 & 0.29 & 0.70 & 0.83 & 305\\[0.25em]
 
 & 1 & 1.00 & 1.00 & 1.00 & 1.00 & 1.00 & 40\\
 & 3 & 0.74 & 0.99 & 1.00 & 1.00 & 1.00 & 89\\
Parabolic & 5 & 0.26 & 0.62 & 0.96 & 1.00 & 0.99 & 186\\
 & 7 & 0.13 & 0.29 & 0.66 & 0.98 & 0.71 & 332\\
 & 9 & 0.12 & 0.17 & 0.36 & 0.77 & 0.40 & 415\\[0.25em]
 
 & 1 & 1.00 & 1.00 & 1.00 & 1.00 & 1.00 & 49\\
 & 3 & 0.91 & 1.00 & 1.00 & 1.00 & 1.00 & 69\\
Sine & 5 & 0.54 & 0.90 & 1.00 & 1.00 & 1.00 & 108\\
 & 7 & 0.28 & 0.60 & 0.95 & 1.00 & 1.00 & 160\\
 & 9 & 0.21 & 0.41 & 0.76 & 0.99 & 0.96 & 238\\[0.5em]
 
    \end{tabular}
    \caption{Like Table \ref{tab:seq_vs_nonseq}, but comparing against the test by \citet{Hoeffding1948}. \label{tab:hoeffdingsd}}
\end{table}

\begin{table}[t]
    \centering
    \small
    \begin{tabular}{llllllcc}
Simulation & $l$ & \multicolumn{4}{c}{Chatterjee's test with $n$ samples} & \multicolumn{2}{c}{Sequential rank test} \\[0.5em]
& & 64 & 128 & 256 & 512 & Power & Mean\\[0.25em]

 & 1 & 0.12 & 0.17 & 0.23 & 0.37 & 1.00 & 100\\
 & 3 & 0.10 & 0.16 & 0.22 & 0.34 & 1.00 & 110\\
Checkerboard & 5 & 0.08 & 0.10 & 0.10 & 0.17 & 1.00 & 195\\
 & 7 & 0.07 & 0.08 & 0.09 & 0.11 & 0.71 & 353\\
 & 9 & 0.06 & 0.06 & 0.08 & 0.08 & 0.35 & 429\\[0.25em]
 
 & 1 & 0.59 & 0.84 & 0.98 & 1.00 & 1.00 & 45\\
 & 3 & 0.25 & 0.39 & 0.61 & 0.84 & 1.00 & 72\\
Circular & 5 & 0.12 & 0.14 & 0.21 & 0.30 & 1.00 & 144\\
 & 7 & 0.06 & 0.06 & 0.10 & 0.11 & 0.68 & 342\\
 & 9 & 0.06 & 0.05 & 0.07 & 0.07 & 0.17 & 467\\[0.25em]
 
 & 1 & 1.00 & 1.00 & 1.00 & 1.00 & 1.00 & 23\\
 & 3 & 0.69 & 0.92 & 1.00 & 1.00 & 1.00 & 59\\
Linear & 5 & 0.22 & 0.38 & 0.56 & 0.82 & 1.00 & 135\\
 & 7 & 0.13 & 0.19 & 0.24 & 0.39 & 0.86 & 252\\
 & 9 & 0.10 & 0.13 & 0.16 & 0.23 & 0.58 & 357\\[0.25em]
 
 & 1 & 0.17 & 0.21 & 0.33 & 0.51 & 1.00 & 121\\
 & 3 & 0.13 & 0.16 & 0.26 & 0.38 & 1.00 & 144\\
Local & 5 & 0.10 & 0.14 & 0.18 & 0.28 & 1.00 & 189\\
 & 7 & 0.10 & 0.14 & 0.17 & 0.22 & 0.95 & 248\\
 & 9 & 0.08 & 0.10 & 0.14 & 0.18 & 0.83 & 305\\[0.25em]
 
 & 1 & 1.00 & 1.00 & 1.00 & 1.00 & 1.00 & 40\\
 & 3 & 0.67 & 0.93 & 1.00 & 1.00 & 1.00 & 89\\
Parabolic & 5 & 0.25 & 0.41 & 0.61 & 0.83 & 0.99 & 186\\
 & 7 & 0.11 & 0.18 & 0.30 & 0.45 & 0.71 & 332\\
 & 9 & 0.11 & 0.12 & 0.18 & 0.24 & 0.40 & 415\\[0.25em]
 
 & 1 & 1.00 & 1.00 & 1.00 & 1.00 & 1.00 & 49\\
 & 3 & 1.00 & 1.00 & 1.00 & 1.00 & 1.00 & 69\\
Sine & 5 & 0.76 & 0.95 & 1.00 & 1.00 & 1.00 & 108\\
 & 7 & 0.40 & 0.66 & 0.89 & 0.99 & 1.00 & 160\\
 & 9 & 0.26 & 0.41 & 0.60 & 0.85 & 0.96 & 238\\[0.5em]
 
    \end{tabular}
    \caption{Like Table \ref{tab:seq_vs_nonseq}, but comparing against the test by \citet{Chatterjee2021}. \label{tab:chatterjee}}
\end{table}

\end{document}